\documentclass[journal,draftclsnofoot,onecolumn,12pt,twoside]{IEEEtranTCOM}
\usepackage{amsmath,amsfonts}
\usepackage{algorithmic}
\usepackage{algorithm}
\usepackage{array}
\usepackage{textcomp}
\usepackage{url}
\usepackage{verbatim}
\usepackage{graphicx}
\usepackage{float}
\usepackage{subfig}
\usepackage{cite}
\usepackage{comment}
\usepackage{amsthm}
\usepackage{bm}
\usepackage[colorlinks=true,linkcolor=blue]{hyperref}
\usepackage{xcolor,etoolbox}
\usepackage{tikz}
\theoremstyle{remark}
\usepackage{pgfplots}
\usetikzlibrary{calc,arrows.meta,positioning}
\usetikzlibrary{shapes.arrows}
\usetikzlibrary{fit,backgrounds} 
\usepackage{standalone}
\usepackage{mathrsfs}
\usepackage{amssymb}
\usepackage[justification=centering]{caption}
\usepackage{mathtools}
\usepackage{thmtools,thm-restate}
\usepackage[square, comma, sort&compress, numbers]{natbib}
\usepackage{enumerate}
\newcommand{\oplabel}[1]{\refstepcounter{equation}(\theequation\ltx@label{#1})}
\newtheorem{theorem}{Theorem}
\newtheorem{lemma}{Lemma}
\newtheorem{example}{Example}
\newtheorem{remark}{Remark}
\newtheorem{corollary}{Corollary}
\newtheorem{definition}{Definition}
\newtheorem{proposition}{Proposition}
\newcommand{\RNum}[1]{\uppercase\expandafter{\romannumeral #1\relax}}
	
\newcommand{\tb}[1]{{\textbf{#1}}}	
\normalsize
\begin{document}
	
	%
	\title{Semantic-Aware Multi-Terminal Coding for Gaussian Mixture Sources}
	%
	%
	%
	
	\author{\thanks{\emph{(Corresponding Authors: Shuo Shao, Yongpeng Wu.)}}
		Yuxuan Shi\thanks{Yuxuan Shi and Shuo Shao are with the School of Cyber and Engineering, Shanghai Jiao Tong University, Shanghai 200240, China (e-mail: ge49fuy@sjtu.edu.cn; shuoshao@sjtu.edu.cn).}, Shuo Shao, \IEEEmembership{Member,~IEEE}, Yongpeng Wu\thanks{Yongpeng Wu and Wenjun Zhang is with the Department of Electronic Engineering, Shanghai Jiao Tong University, Shanghai 200240, China (e-mail: yongpeng.wu@sjtu.edu.cn, zhangwenjun@sjtu.edu.cn).}, \IEEEmembership{Senior Member,~IEEE}, Jun Chen\thanks{Jun Chen is with the Department of Electrical and Computer Engineering, McMaster University, Hamilton, ON L8S 4K1, Canada (e-mail: 			junchen@ece.mcmaster.ca).}, \IEEEmembership{Senior Member,~IEEE}
		Wenjun Zhang, \IEEEmembership{Fellow,~IEEE}}

\maketitle

\begin{abstract}
	A novel distributed source coding model which named semantic-aware multi-terminal (MT) source coding is proposed and investigated in the paper, where multiple agents independently encode an imperceptible semantic source, while both semantic and observations are reconstructed within their respective fidelity criteria. We start from a generalized single-letter characterization of sum rate-distortion region of this problem. Furthermore, we propose a mixed MSE-Log loss framework for this model and specifically depict the rate-distortion bounds when sources are Gaussian mixture distributed. For this case, we first present a relative tight outer bound and explore the activeness of semantic and observation distortion constraints, in which we find that good observation reconstruction will not incur too much semantic errors, but not vice versa. Moreover, we provide a practical coding scheme functioning as an achievable regime of inner bound with the performance analysis and simulation results, which verifies the feasibility of the idea "detect and compress" for Gaussian mixture sources. Our results provide theoretical instructions on the fundamental limits and can be used to guide the practical semantic-aware coding designs for multi-user scenarios.
\end{abstract}


%
\IEEEpeerreviewmaketitle

\section{Introduction}
\IEEEPARstart{I}{n} nowadays sixth-generation (6G) networks, a new communication paradigm named semantic communication, emerges with superiorities on high compression ratio and transmission efficiency in comparison with the conventional ones. This is mainly attributed to the consideration of information meanings and intents, which is always neglected in existing frameworks. Being aware of the correlation among semantics, this new paradigm aims to reduce the source redundancy and increase transmission robustness. For this very reason, enormous number of works are emerging to investigate the potential of semantic communication both theoretically and practically.

The conception of communication over semantic level was first presented in the masterpiece of Shannon \cite{Burks_Shannon_Weaver_1951} in 1948. Hereafter, Bar-Hillel and Carnap \cite{Carnap1953An} revisited the omitted issues in his work and defined semantic information with the help of truth table preliminarily. Bao et al. \cite{Bao2011} emphasized the role of background information in semantic-aware scenarios, while Guler et al. \cite{Guler2016} proposed a generalized framework which minimizes erroneous semantics. The aforementioned pioneering works in the past seven decades mainly focus on characterizing semantics among texts via logical probability, resulting from the difficulty to describe semantics of multi-modal data. Fortunately, this difficulty has been overcome in recent years due to the prosperity in deep learning techniques involving natural language processing (NLP), computer vision, etc. For practical application, numerous semantic-aware frameworks on images, audios and videos are proposed and attain outstanding performances in terms of different tasks (see e.g. \cite{Nariman2018,Xie_Qin_Li_Juang_2020,Zhenzi2021,Kountouris_Pappas_2021,Huang2021,Dommel_Utkovski_Simeone_Stanczak_2021,Jiang_Wen_Jin_Li_2021} for some representative works).

Besides, a few works concentrate on the fundamental limits of these semantic-aware scenarios. A heuristic model based on the indirect source coding problem was proposed by Liu et al. in \cite {Liu_Shao_Zhang_Poor_2022}, in which the visible source also required recovery in their model, and the authors referred to this hidden source as semantic source. Furthermore, Shi et al. \cite{Shi2023} extended the model and gave an excess distortion exponent analysis for joint source-channel coding scheme for semantic communications. In general, the introduction of an extra constraint inevitably results in differences of theoretical analysis, e.g. the reproof of the single letter characterization of rate-distortion function (RDF). Besides, multiple distortion constraints complicates the analysis of RDF with specific source distributions. For this reason, Stavrou and Kountouris \cite{Stavrou22} developed a semantic-aware Blahut-Arimoto (BA) algorithm for the computation of intrinsic or extrinsic sources under arbitrary discrete distributions. Nevertheless, the above works mainly consider single user case, which is restrictive for most practical situations involving interactions among multiple users. Moreover, these works always make assumption that semantic source is defined over infinite alphabet for convenience, which is not in accordance with the situations in real world, such as segmentation and classification tasks. This hence motivates a multi-terminal (MT) model for semantic-aware source coding, in which we make proper assumptions on semantic alphabet and distortion measures. 


Inspired by \cite{Liu_Shao_Zhang_Poor_2022}, we aim to investigate an MT source coding scenario characterizing semantic information. In this scenario, an invisible semantic source is set up, while multiple agents observe the same semantic source and encode independently. Different from the classic problems, all semantics and observations are required to reconstruct within their fidelity criteria. This modeling is mainly motivated by some widely-used semantic-aware scenarios, which intend to reconstruct the original media meanwhile execute the downstream tasks, e.g. the low bit-rate video understanding \cite{Tian2022}, the semantic-to-signal scalable image compression \cite{Liu_Liu_Li_Yan_Li_2021}. For further analysis, we assume semantic source defined over a finite alphabet, meanwhile each visible source follows a Gaussian distribution by observing the semantic source. Overall the visible sources follow a vector Gaussian mixture distribution. The basis for this modeling can be primarily attributed to the following two reasons: First, defining a semantic source over a finite alphabet is typically more reasonable than over an infinite one; Second, a Gaussian mixture has the capability to approximate a wide range of irregular distributions, including those encountered in real-world datasets. Under this source coding scenario, an essential problem is to characterize the rate-distortion region when fixing the coding rate and reconstruction strategy.

Based on the discussion, we introduce a mathematical model named as semantic-aware MT source coding in this paper. Specifically, we start from an invisible semantic source $S$, and $L$ agents with different observation $X_1,X_2,\cdots,X_L$. The decoder is required to reconstruct both the semantic source $\hat{S}$ and the observations $\hat{X}_1,\hat{X}_2,\cdots,\hat{X}_L$ without exceeding the distortion constraints, respectively. An intuitive interpretation of this scenario is a Chief Executive Officer (CEO) problem \cite{Oohama_1997,Viswanathan_Berger_1997,Wang2010,Yang_Xiong_2012} coupled with an MT source coding problem \cite{Berger_Yeung_1989,Oohama_1998,Wagner_Anantharam_2008,Wang2010,Wagner_2008}. To investigate the RDF of this model, we follow the basic steps of characterizing RDF bounds from the generalized form with arbitrary sources to an analysis-friendly statement with specific sources. Note that the dilemmas arise from the coexistence of multiple constraints, as outlined below: for the single letter characterization, decoupling methods for arbitrarily correlated sources are difficult to obtain. Moreover, for the specific analysis with Gaussian mixture, it is challenging to describe the mutual information between discrete semantics and continuous codewords when formulating converse bound. Meanwhile, the above challenge prevents us from constructing optimal codebook for an achievable scheme via random binning, in view of the inconsistency among sources alphabets. In the following, we conclude the main contributions of this work.
\begin{enumerate}[a)]
	\item We propose a general distributed source coding model on characterizing the semantic-aware MT problem, which aligns with reality and provides guiding significance for the AI-based frameworks. Based on the model, we present inner and outer bounds of RDF in the single letter form, in which we assume a Markov coupled random variable for the semantic and observed sources decorelation. Moreover, degenerating cases of our bounds are also discussed, which verifies that our generalized bounds can cover existing works.
	\item For further analysis, we propose a mixed MSE-logarithmic (MML) loss framework, by considering practical distortion measures for semantic-aware scenarios, where logarithmic loss measures the semantics and mean square error (MSE) measures the observations. By specifying that sources are Gaussian mixture distributed, we present a relatively tight converse bound on sum-rate distortion in the form of an optimization problem, where the mutual information between discrete semantics and continuous codewords can be connected to the error probability via Fano's inequality. We further discuss the activeness of semantic and observation distortion constraints to completely solve the optimization problem of RDF, which unveils an intriguing observation that semantic distortion can be upper bounded when the agents error is fixed, but the reverse is not true.
	\item Furnished with outer bound, we also provide a practical coding scheme for the semantic-aware MT problem with Gaussian mixture sources, which does not rely on the random binning and method of type techniques. This coding scheme comprises clusters, compressors, and Slepian-Wolf encoders, where each observation undergoes clustering before quantization, namely following "detect and compress" idea. Furthermore, we offer an inner bound based on the coding scheme, and the numerical results show its superiority in comparison with existing methods. Among these, we notice that there though exists gap between outer and inner bounds in our scenario, the gap will be eliminated with the increase of the ratio of signal to noise.
\end{enumerate}
This paper is organized as follows: in Section \ref{Sec2}, we introduce the notations on semantic-aware MT source coding problem, including the system model and definitions. In Section \ref{Sec3}, we show a Berger-Tung bounds-based rate-region in  semantic-aware MT problem. In Section \ref{Sec4}, we propose a novel framework named as mixed MSE-Log loss, and derive bounds on sum rate-distortion when sources are Gaussian mixture distributed. In Section \ref{Sec5}, a practical coding scheme is designed for approaching theoretical performance. Finally the Section \ref{Sec6} concludes the paper.

Throughout the paper, random variables are represented by upper case letters, whose realizations and alphabets are written in lower case and calligraphy, respectively. For instance, $x$ is the realization of random variable $X$ and picks values in $\mathcal{X}$. The cardinality of a set $\mathcal{A}$ is $|\mathcal{A}|$, and $\mathcal{A}/a$ denotes the set $\mathcal{A}$ excluding the element $a$. We abbreviate the tuple $(X_{1},X_{2},\cdots,X_{n})$ as $X^n$ and the realization $x^n$ follows similarly. $\mathbb{E}[\cdot]$, $\mathbb{V}[\cdot]$, $h(\cdot)$ and $H(\cdot)$ are expectation, variance, differential entropy and discrete entropy functions respectively. More specifically, $H_2(\cdot)$ is the binary entropy function, and $p*q=p(1-q)+q(1-p)$ denotes the binary convolution operation. Moreover, an italic and bold symbol $\bm{X}$ represents a vector and bold symbol $\tb{X}$ stands for a matrix. $\bm{0}_L$ and $\bm{1}_L$ are all zero and one vector, respectively. $\bm{e}_i$ denotes all zero vector except for $1$ at its $i$-th position. Specifically, $(\cdot)^T$, $\mathrm{tr}(\cdot)$ and $\det(\cdot)$ denote the transpose, trace and determinant operators, respectively. Besides, $\mathrm{diag}\{\cdot\}$ represents a diagonal matrix with only non-zero elements on its principal diagonal.  $\tb{O}$ is all-zero matrix (not necessarily square matrix) and $\tb{I}_L$ means $L\times L$ identity matrix. For further use, $\mathcal{N}(\bm{x};\bm{\mu},\tb{K})$ denotes the probability density function of an $L$-length vector Gaussian random variable with realization $\bm{x}$, namely
\[
\mathcal{N}(\bm{x};\bm{\mu},\tb{K})=\frac{1}{(2\pi)^{L/2}\det(\tb{K})^{1/2}}\exp\left\{-\frac{1}{2}(\bm{x}-\bm{\mu})^T\tb{K}^{-1}(\bm{x}-\bm{\mu})\right\}.
\]

\section{Problem Formulation}	\label{Sec2}
In this section, the system model of the semantic-aware MT source coding is presented. We also introduce the definitions of semantic-aware MT source coding and the sum rate-distortion region. 
\subsection{Problem Formulation}
\begin{figure}[h]
	\centering
	\includegraphics[width=0.8\textwidth]{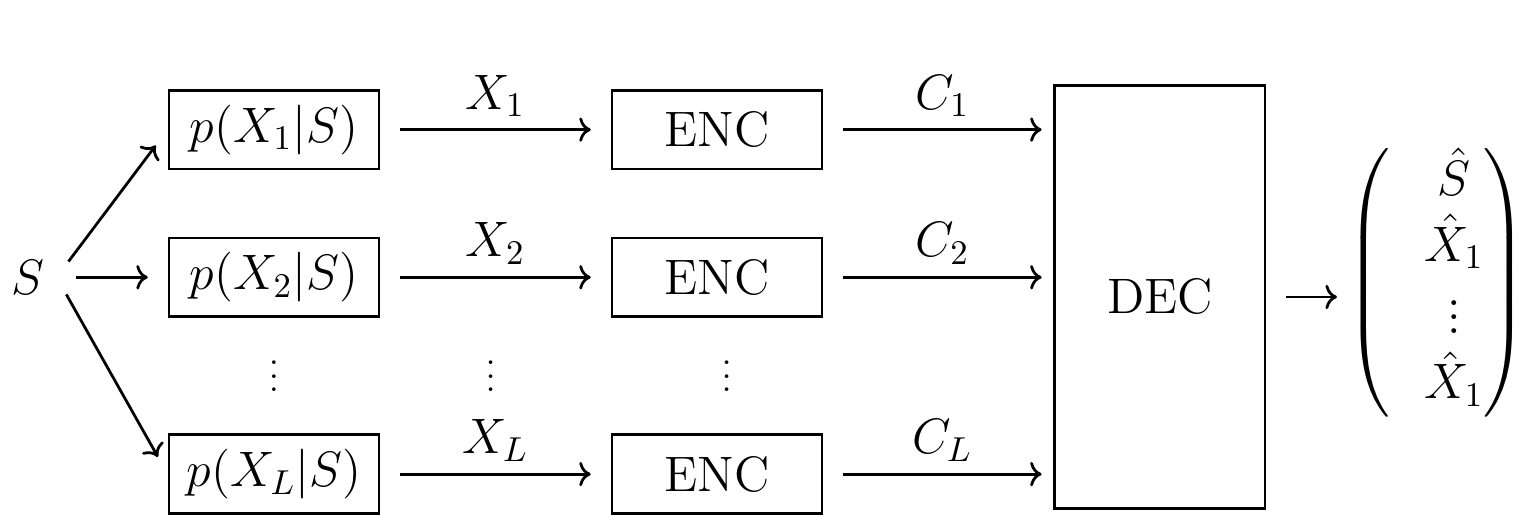}
	\caption{A semantic-aware multi-terminal source coding problem}
	\label{SMT}
\end{figure}
A generalized semantic-aware multi-terminal source coding problem is depicted in Fig.\ref{SMT}. We assume a memoryless information source produces independent and identical (i.i.d) random variables $S$ with probability distribution $P_{S}$ over alphabet $\mathcal{S}$. We interpret source $S$ as the semantic features of events/objects and cannot be observed directly. Meanwhile $L$ agents obtain corrupted observations of the semantic source, in which each observation $X_i$ takes value in $\mathcal{X}_i$ for $i=1,2,\cdots,L$. These observations are encoded independently and decoded together. Notably, the main difference between the semantic-aware MT source coding model and the classic CEO problem, is that the decoder intends to reconstruct not only the semantics $\hat{S}$ but also the observations $\hat{X}_1,\cdots,\hat{X}_L$.

By incorporating the $k$-length block coding setting into the model, the information to be compressed is no longer a scalar but indeed a vector. More specifically, for $i=1,2,\cdots,L$, the $i$-th encoder is defined as a mapping from the $k$-fold Cartesian products of alphabets $\mathcal{X}^k$ to the codewords alphabet: 
\begin{align}
	\varphi_i(\cdot): \mathcal{X}^k\mapsto\mathcal{C}_i=\{1,2,\cdots,|\mathcal{C}_i|\},\qquad i=1,2,\cdots,L\notag.
\end{align} 
Consequently, the unique decoder is composed of $L+1$ mappings as:
\begin{align}
	\psi_S(\cdot):\prod_{i=1}^L\mathcal{C}_i\mapsto\hat{\mathcal{S}}^k;\psi_{X_i}(\cdot):\prod_{i=1}^L\mathcal{C}_i\mapsto\hat{\mathcal{X}}^k,\quad i=1,2,\cdots,L\notag,
\end{align}
where $\hat{\mathcal{S}}$ and $\hat{\mathcal{X}}$ are the alphabets of 
$\hat{S}$ and $\hat{X}_i$, for $i=1,2,\cdots,L$. Consequently the code rate at $i$-th agent and the sum rate are defined as 
\begin{align}
	R_i=\frac{1}{k}\log|\mathcal{C}_i|\notag,\qquad R_\mathrm{sum}=\sum_{i=1}^LR_i.
\end{align}
For simplicity, we write $\bm{X}=(X_1,X_2,\cdots,X_L)^T$, $\bm{X}^k=((X_1^k)^T,(X_2^k)^T,\cdots,(X_L^k)^T)^T$, the encoder tuple $\bm{\varphi}=(\varphi_1(\cdot),\varphi_2(\cdot),\cdots,\varphi_L(\cdot))^T$, decoder tuple $\bm{\psi}_X=(\psi_1(\cdot),\psi_2(\cdot),\cdots,\psi_L(\cdot))^T$ and $C_{\mathcal{L}}=\{C_1,C_2,\cdots,C_L\}$. Moreover, by $\mathcal{F}^k_L(R)$ we denote all the encoder-decoder pairs $(\bm{\varphi},\psi_S,\bm{\psi}_X)$. Now for $\hat{S}^k=\psi_S(C_{\mathcal{L}})$ and $\hat{\bm{X}}^k=\bm{\psi}_{X}(C_{\mathcal{L}})$, and two distortion measures $d_S:\mathcal{S}\times\hat{\mathcal{S}}\mapsto\mathbb{R}^+$ and $d_X:\mathcal{X}\times\hat{\mathcal{X}}\mapsto\mathbb{R}^{+}$, we write the block-wise distortion measure functions on semantic and observations,  respectively:
\begin{align}
	&d_S^k\left(S^{k}, \hat{S}^{k}\right)=	\frac{1}{k}\sum_{j=1}^kd_S(S_j,\hat{S}_j),\label{1}\\
	&d_X^k\left(X_i^{k}, \hat{X}_i^{k}\right)=\frac{1}{k}\sum_{j=1}^kd_X(X_{i,j},\hat{X}_{i,j}),\quad i=1,\cdots,L.\label{2}
\end{align}
\subsection{Problem Formulation}
Equipped with the above definition, we are able to describe the admissible rate region of semantic-aware multi-terminal source coding problem. Moreover, we state the definition of inactive (or we say dummy) constraint in our multi-constraint scenario.
\begin{definition}\label{admissible}
	Given a fixed positive integer $L$ and $k$, a rate-distortion pair $(R, D_S,\bm{D}_X)$ composed of total rate constraint $R=R_\mathrm{sum}$, semantic distortion $D_S$ and observation distortion vector $\bm{D}_X=\left( D_{X_1},D_{X_2}\cdots,D_{X_L}\right)$ is admissible if there exists a pair $(\bm{\varphi},\psi_S,\bm{\psi}_X) \in \mathcal{F}_{L}^{(k)}(R)$ such that 
	\begin{align}
		&\mathbb{E}d_S^k\left(S^{k}, \hat{S}^{k}\right) \leq D_S, \label{Ds}\\
		&\mathbb{E}d_X^k\left(X_i^{k}, \hat{X}_i^{k}\right)\leq D_{X_i},\qquad i=1,\cdots,L\label{Dx}.
	\end{align}
\end{definition}
\noindent With this definition, the RDF for semantic-aware MT problem can be formulated as
\begin{align}
	&R^{(k)}(D_S,\bm{D}_X)\triangleq\{R:(R,D_S,\bm{D}_X) \text { is admissible}\},\\
	&R(D_S,\bm{D}_X)=\inf_{k\geq1}R^{(k)}(D_S,\bm{D}_X).\label{6}
\end{align}
The next definition formulates the activeness of distortion constraints in our semantic-aware scenario.
\begin{definition}\label{dummy}
	For the function of $R(D_S,\bm{D}_X)$ in Eq. \eqref{6}, with $H(S|X)\leq D_S\leq H(S)$ and $0\leq \sum_{i=1}^LD_{X_i}\leq H(\bm{X}|S)$, we define that
	\begin{enumerate}
		\item the semantic distortion constraint in Eq. \eqref{Ds} is inactive for the rate distortion function $R(D_S,\bm{D}_X)$, if we can find non-negative $\bm{D}_{X}^\star$ such that $\forall \bm{D}_{X}\prec \bm{D}_{X}^\star$, there exists $0<\Delta_S$ satisfying
		\begin{align}
			&R(D_S+\Delta_S,\bm{D}_X)=R(D_S,\bm{D}_X).
		\end{align}
		\item the observation distortion constraint in Eq. \eqref{Dx} is inactive for the rate distortion function $R(D_S,\bm{D}_X)$, if we can find non-negative $D_S^\star$ such that $\forall D_S< D_S^\star$, there exists L-length vector $\bm{0}\prec\bm{\Delta}_X$ satisfying
		\begin{align}
			&R(D_S,\bm{D}_X+ \bm{\Delta}_X)=R(D_S,\bm{D}_X).
		\end{align}
		where $\bm{A}=(A_1,\cdots,A_L)\prec\bm{B}=(B_1,\cdots,B_L)$ denotes the vector inequality that $A_i\leq B_i$ for $i=1,\cdots,L$ and $|\bm{A}|<|\bm{B}|$.
	\end{enumerate}
\end{definition}
Equipped with Definition \ref{dummy}, the activeness of different distortion constraints can be discussed, and thus we can simplify the formulation of RDF of our semantic-aware MT problem by arguing the existence of $D_S^\star$ and $\bm{D}_{X}^\star$ for specific RDF, respectively. Take $R(D_S,\bm{D}_X)$ for instance, we can assert that the semantic distortion constraint of $R(D_S,\bm{D}_X)$ is dummy when $D_S\leq D_S^\star$, through finding a $D_S^\star$. 

With the above definitions, in this paper, we mainly aim to characterize the RDF $R(D_S,\bm{D}_X)$, and investigate its behaviors.

\section{Sum Rate-Distortion Characterization of Semantic-aware MT Problem in Single Letter Form} \label{Sec3}
In this section, we present outer and inner bounds for the RDF of semantic-aware MT source coding problem for arbitrary source distributions. After that, the generalized statements are verified to cover some existing degenerated cases, including classic CEO, MT problems.
\subsection{Bounds Characterization of Semantic-aware MT Problem}\label{3.2}
In the following we first present specific definitions of sum rate-distortion bounds.
\begin{definition}\label{Definition_Cout}
	Given non-negative distortions $D_S,\bm{D}_X$, we define
	\begin{align}
		R_\mathrm{out}(D_S,\bm{D}_X)\triangleq&\max_{\bm{Y},\bm{W}}\min_{\bm{U}} \left\{I(\bm{Y};\bm{U})+\sum_{i=1}^LI(X_i;U_i|\bm{Y},\bm{W})\right\},\label{Outer_bound_single_letter}\\
		&\text{s.t.}\hspace{0.1cm}\mathbb{E}d_S\left(S;f(\bm{U})\right)\leq D_S,\label{ds}\\ &\hspace{0.6cm}\mathbb{E}d_X\left(X_i;g_i(\bm{U})\right)\leq D_{X_i},\hspace{0.1cm} i=1,\cdots,L,\label{dx}
	\end{align}
	for a joint distribution $P_{S\bm{Y}\bm{X}\bm{W}\bm{U}}$ of the form
	\begin{align}
		&P(\bm{y}|\bm{x})P(s)P(w)\prod_{i=1}^LP(x_i|s)P(u_i|x_i,w)\label{jointd}.
	\end{align}
	Herein $\mathcal{Y}$ denotes the set including all random variables that $X_1,X_2,\cdots,$ $X_L$ are conditional independent\footnote{It can be easily verified $\mathcal{Y}$ is nonempty since it contains $\bm{Y}=\left\{X_1,X_2,\cdots,X_{L-1}\right\}$. Moreover, by fixing a certain coupling method, the cardinality of set $\bm{Y}$ can be considered as finite.} if given $\bm{Y}$. Besides, $\bm{U}=\{U_1,U_2,\cdots,U_L\}$, $U_i\in\mathcal{U}_i$ while $|\mathcal{U}_i|\leq|\mathcal{X}_i|+2^L+L-2$ for $i=1,2,\cdots,L$, and the reproduction functions are defined as
	\begin{align}
		f(\cdot):\prod_{i=1}^L\mathcal{U}_i\mapsto \hat{\mathcal{S}},\qquad
		g_i(\cdot):\prod_{i=1}^L\mathcal{U}_i\mapsto \hat{\mathcal{X}}_i,\quad \quad i=1,\cdots,L.\notag
	\end{align}
\end{definition}
\begin{definition}\label{Definition_Cinn}
	Given non-negative distortions $D_S,\bm{D}_X$, we define $R_\mathrm{in}(D_S,\bm{D}_X)$ as 
	\begin{align}
		R_\mathrm{in}(D_S,\bm{D}_X)\triangleq&\min_{\bm{V}} I(\bm{X};\bm{V}),\label{Inner_bound_single_letter}\\
		&\text{s.t.}\hspace{0.1cm}\mathbb{E}d_S\left(S;f^\star(\bm{V})\right)\leq D_S,\label{z9}\\ &\hspace{0.6cm}\mathbb{E}d_X\left(X_i;g^\star_i(\bm{V})\right)\leq D_{X_i},\hspace{0.1cm} i=1,\cdots,L,\label{z10}
	\end{align}
	for a joint distribution $P_{S\bm{X}\bm{V}}$ of the form
	\begin{align}
		P(s)\prod_{i=1}^LP(x_i|s)P(v_i|x_i),\label{166}
	\end{align} 
	where $\bm{V}=\{V_1,V_2,\cdots,V_L\}$ and $V_i\in\mathcal{V}_i$ and reproduction functions
	\begin{align}
		f^\star(\cdot):\prod_{i=1}^L\mathcal{V}_i\mapsto \hat{\mathcal{S}},\qquad
		g^\star_i(\cdot):\prod_{i=1}^L\mathcal{V}_i\mapsto \hat{\mathcal{X}}_i,\quad i=1,\cdots,L.\notag
	\end{align}
\end{definition}
\noindent
Equipped with above definitions, the sum-rate within the required distortions of semantic-aware MT problem can be bounded as follows.
\begin{theorem}\label{Theorem_label_outer_bound}
	If $(R,D_s,\bm{D}_X)$ is admissible, then
	\begin{align}
		R(D_S,\bm{D}_X)\geq R_\mathrm{out}(D_S,\bm{D}_X)\label{outer_bound_single_letter}
	\end{align}
\end{theorem}
\begin{proof}
	See Appendix \ref{proof_single_letter_outer_bound}.
\end{proof}
\begin{proposition}\label{Theorem_label_inner_bound}
	\begin{equation}
		R(D_S,\bm{D}_X)\leq R_\mathrm{in}(D_S,\bm{D}_X).\label{inner_bound_single_letter}
	\end{equation}
\end{proposition}
\begin{proof}
	The proposition can be readily verified with standard Berger-Tung inner bound \cite{Berger1996} with an extra semantic distortion constraint, hence we omit the proof here.
\end{proof}

\subsection{Some Degenerated Cases}
The degenerated cases of Theorem \ref{Theorem_label_outer_bound} and Prop. \ref{Theorem_label_inner_bound} are stated as follows. We first introduced some definitions of the existing bounds.
\begin{definition}\label{Def5} Given non-negative $D,D_S,D_X$ and $L$-length vector $\bm{D}$, by slightly modifying the notations in existing works, we define
	\begin{enumerate}
		\item single user rate-distortion function characterizing semantic information \cite{Liu_Shao_Zhang_Poor_2022} as
		\begin{align}
			\widehat{R}(D_S,D_X)\triangleq&\min_{P_{\hat{S}\hat{X}|X}}I(X;\hat{S},\hat{X})\notag\\
			&\text{s.t.}\quad\mathbb{E}[d(X,\hat{X})]\leq D_X,\notag\\
			&\hspace{0.95cm}\mathbb{E}[\bar{d}(X,\hat{S})]\leq D_S.\notag
		\end{align}
		where $S-X-(\hat{S},\hat{X})$ forms a Markov chain, and $\bar{d}(x,\hat{s})=\sum_{s}p(s|x)d_S(s,\hat{s})$.
		\item outer and inner bounds of CEO problem\footnote{Herein we restate the formulation in \cite{Courtade_Weissman_2014} without the time sharing random variable since we only care about the sum rate in this paper.} with logarithmic loss  (c.f. \cite{Courtade_Weissman_2014}) as
		\begin{align}
			R^\mathrm{CEO}_{\mathrm{out}}(D)=&\min_{U_1,U_2}\left[I(U_1;X_1|S)+I(U_2;X_2|S)+H(S)-D\right]^+,\notag\\
			&\text{s.t.}\quad D\geq H(S|U_1,U_2)\notag
		\end{align}
		according to a joint distribution \[
		P(s)P(x_1|s)P(x_2|s)P(u_1|x_1)P(u_2|x_2),
		\]
		and 
		\begin{align}
			R^\mathrm{CEO}_{\mathrm{in}}(D)=&\min_{U_1,U_2}I(U_1,U_2;X_1,X_2)\notag\\
			&\text{s.t.}\quad D\geq H(S|U_1,U_2)\notag
		\end{align}
		according to a joint distribution\footnote{For simplicity we follow the same statements as Courtade and Weissman \cite{Courtade_Weissman_2014}, while the auxiliary random variables in outer and inner bounds will not be the same one in general.} \[
		P(s,x_1,x_2)P(u_1|x_1)P(u_2|x_2).
		\]
		\item Berger-Tung outer bound of MT problem (c.f. \cite{Berger1978}) as
		\begin{align}
			R_{\mathrm{out}}^{\mathrm{BT}}(\bm{D})=&\min_{\bm{U}}I(\bm{X};\bm{U})\notag\\
			&\text{s.t.}\quad\mathbb{E}[d(X_i,\hat{X}_i)]\leq D_i,\quad i=1,2,\cdots,L,\notag
		\end{align}
		where $\bm{D}=\{D_1,\cdots,D_L\}$ and the Markov chain behaves $U_i\rightarrow X_i\rightarrow \bm{X}_{i^c}$ for $i=1,2,\cdots,L$.
	\end{enumerate}
\end{definition}
Now with Definition \ref{Def5}, we clarify the degeneration of our bounds in the following corollary.
\begin{corollary}\label{corollary_degeneration}
	If $(R,D_s,\bm{D}_X)$ is admissible,
	\begin{enumerate}
		\item then by fixing $L=1$ and non-negative $D_S,D_X$,
		\begin{align}
			R_{\mathrm{in}}(D_S,D_X)=R(D_S,D_X)=R_{\mathrm{out}}(D_S,D_X)=\widehat{R}(D_S,D_X),\label{corollaryp2p}
		\end{align}
		where Eq. \eqref{outer_bound_single_letter} coincides Eq. \eqref{inner_bound_single_letter} and they are reduced to the indirect rate-distortion characterization for semantics.
		\item 
		then by fixing $L=2$, $d_S(s,\hat{s})=-\log(\hat{s}(s))$, i.e. the distortion is specified to the logarithmic loss, and $D\geq0$,
		\begin{align}
			R^\mathrm{CEO}_{\mathrm{in}}(D)=R_{\mathrm{in}}(D,+\infty)=R(D,+\infty)=R_{\mathrm{out}}(D,+\infty)= R^\mathrm{CEO}_{\mathrm{out}}(D),\label{corollaryCEO}
		\end{align}
		where Eq. \eqref{outer_bound_single_letter} and Eq. \eqref{inner_bound_single_letter} coincide to the two-user CEO problem with logarithmic loss. Herein $R(D,+\infty)$ denotes the sum rate defined in Eq. \eqref{6} by relaxing $D_{X_i}\rightarrow\infty$ for all $i$.
		\item 
		then for non-negative $\bm{D}$,
		\begin{align}
			R_{\mathrm{out}}(+\infty,\bm{D})\geq R^\mathrm{BT}_{\mathrm{out}}(\bm{D}),\label{corollaryMT}
		\end{align}
		where Eq. \eqref{outer_bound_single_letter} is degenerated to the Berger-Tung outer bound. Herein $R_{\mathrm{out}}(+\infty,\bm{D}_X)$ denotes the sum rate defined in Eq. \eqref{6} by relaxing $D_{S}\rightarrow\infty$.
	\end{enumerate}
\end{corollary}
\begin{proof}
	See Appendix \ref{proof_corollary_degeneration}.
\end{proof}

In this subsection we present three degenerated cases corresponding to the point-to-point compression, CEO problem and MT problem, respectively. The corollary verify that our inner and outer bounds are more generalized results which can cover existing conclusions in \cite{Liu_Shao_Zhang_Poor_2022,Courtade_Weissman_2014,Berger1978}. Besides, from Corollary \ref{corollary_degeneration} we also observe a known conclusion that the converse and achievability bounds in point-to-point compression and CEO problem always meet, while the coincidence in MT problem does not exists in general. In our semantic-aware MT system, this non-coincidence still exists generally. In the following, we will present a specific theoretical model to analyze the sum-rate performance.

\section{A Novel Model: Mixed MSE-Log Loss Semantic-Aware MT Problem and Its Fundamental Limits with Gaussian Mixture Sources}\label{Sec4}

\subsection{The Proposed Mixed MSE-Log Loss Framework}
In most AI-based frameworks, the quadratic loss is commonly used in regression tasks but restricted being in cope with tasks like classifications and segmentations. Instead of MSE loss, logarithmic loss that measures the distance between distributions is widely used in a larger range of tasks, which is defined as:
\begin{align}
	d\left(x,\hat{x}\right)\triangleq\frac{1}{\log\hat{x}(x)},\label{16}
\end{align}
where $\hat{x}(x)$ denotes the Kullback-Leibler (KL) divergence between the reconstruction $\hat{x}$ and the empirical distribution. A more common version of logarithmic loss appears with expectation over a different distribution and is well known as cross entropy loss. The superiority of this loss is mainly attributed to that its update will not depend on the derivatives of sigmoid outputs, resulting in a faster convergence speed and better accuracy than MSE loss in training progress \cite{Wang2022}\cite{Lorenzo2023}. Note that it is more reasonable to measure the abstract semantic information via distributions but not their actual realizations. Therefore, in this section, we propose a \textbf{mixed MSE-Log (MML) loss framework} for semantic-aware MT source coding problem as following: we apply logarithmic loss on semantic source and MSE loss on observations $\bm{X}$. This is due to the difficulty to measure the invisible semantic information, consequently resulting in a proper distribution measure on semantics. More specifically, considering a slightly modified decoder $\psi_S(\cdot):\prod_{i=1}^L\mathcal{C}_i\mapsto\mathcal{P}(\mathcal{S}^k)$, where the reproduction alphabet $\mathcal{P}(\mathcal{S}^k)$ is the set of probability distribution over $\mathcal{S}^k$, we specify symbol-wised distortion measurements as $d_S(s,\hat{s})=-\log\hat{s}(s)$ and $d_X(x_i,\hat{x}_i)=\left\Vert x_i-\hat{x}_i\right\Vert^2$ for $i=1,2,\cdots,L$, in which we abuse the notation $d_S(\cdot,\cdot):\mathcal{S}\times\mathcal{P}(\mathcal{S})\mapsto\mathbb{R}$.
\begin{remark}
	As mentioned above, our semantic-aware MT problem is an extension of the single user rate-distortion characterizing semantic information from Liu et al. \cite{Liu_Zhang_Poor_2021}, in which they investigated the case that semantic and observations are joint Gaussian distributed. We note that the Gaussian modeling on semantic information is an ideal assumption in most practical scenarios though it presents the possibility to analyze the tight/exact characterization of rate-distortion behavior. However, in most tasks, e.g. clusterings, semantic segmentations or object detection, the cardinality of semantic labels $|\mathcal{S}|$ is actually finite. This motivates us to consider the above model where multi-users obtain continuous observations of a discrete semantic source. The main results of this work lie as follows.
\end{remark}
\subsection{Fundamental Limits of Semantic-Aware MT Problem with Gaussian Mixture Sources}\label{fundamental_limits}
Based on the problem formulation in Sec. \ref{Sec2} and the above remark, we consider such a case that semantic source $S$ is a discrete random variable over alphabet $\mathcal{S}=\{1,2,\cdots,M\}$, and $P\{S=\ell\}=\omega_{\ell}$ for $\ell\in\mathcal{S}$. Meanwhile the observations $\bm{X}$ are highly dependent on $S$ while it is a conditional Gaussian vector with probability density function
\begin{align}
	p_{\bm{X}}(\bm{x}|S=\ell)=\mathcal{N}(\bm{x};\ell\cdot\bm{1},\tb{K}_X)
\end{align}
where $\ell\cdot\bm{1}=(\ell,\ell,\cdots,\ell)$ and $\tb{K}_X$ is a positive definite matrix. In view of the lack of conclusions on Gaussian mixture models in a multi-user setting, we first present some known results on point-to-point rate-distortion function with Gaussian mixture, denoted as $R(D)$.
\begin{proposition}[c.f. \cite{Weidmann_Vetterli_2012} (Conditional rate-distortion bound)]\label{prop_CRDF}
	Given a Gaussian mixture source $X$ with a hidden discrete memoryless source $S$ that switches between $|\mathcal{S}|$ Gaussian sources $\mathcal{N}(0,\sigma_X^2)$ with probabilities $\omega_i=\mathbb{P}\{S=s\}$, RDF can be lower bounded by
	\begin{align}
		R(D)\geq\min_{p(\hat{x}|x,s):\mathbb{E}[X-\hat{X}]\leq D}I(X;\hat{X}|S)=\frac{1}{2}\log\frac{\sigma_X^2}{D}\label{29}.
	\end{align}
	where $\hat{X}$ denotes the reconstruction of $X$.
\end{proposition}
Weidemann and Vetterli investigated the distortion rate behaviors for Gaussian mixtures via Eq. \eqref{29}. Note that this bound is the conditional rate-distortion to encode $X$ with side information $S$ accessing to both compressor and decompressor, thus it is a trivial outer bound if we extend it to our MT problem. However, one will not pleased to see that since this bound directly discards the mutual information $I(S;\hat{X})$, which inevitably entails significant gap when $S$ can be well estimated or alphabet $\mathcal{S}$ is large. To improve this, the following work simply analyzed this mutual information involving discrete and continuous random variables to improve the RDF of Gaussian mixtures under this case.
\begin{proposition}[c.f. \cite{Reznic_Zamir_Feder_2002} (Shannon lower bound)]
	With the same setups in Prop. \ref{prop_CRDF}, and $e\triangleq\frac{1}{2}\min_{i\neq j}|s_i-s_j|$ we have
	\begin{align}
		R(D)\geq H(\bm{\omega})-\epsilon(\sigma)+\frac{1}{2}\log\frac{\sigma_X^2}{D},\label{24}
	\end{align}
	where $\bm{\omega}=(\omega_1,\cdots,\omega_{|\mathcal{S}|})$, $\epsilon(\sigma)= H(S|X)$ and $\sigma\triangleq\frac{e}{\sigma_X^2}$.
\end{proposition}
Reznic et al. started from the Shannon lower bound (SLB) and finally formulated $\epsilon(\sigma)$ via Fano's inequality and hypothesis testing. By comparing the above two inequalities, one can find though SLB is a better outer bound than conditional rate-distortion bound, it cannot be stated in a closed form, since the characterization of Gaussian mixture entropy is still an open question in general. Moreover, we know SLB is tight for point-to-point case however it is not the same case for MT scenarios. Last but not least, our semantic-aware problem intrinsically constrains the process of semantic reconstruction, necessitating a more sophisticated analysis of the mutual information $I(S;\hat{X})$. All these reasons encourage us to explore an improved outer bound for Gaussian mixture sources under the semantic-aware MT scenario. Now let $R^\mathrm{C}(D_S,\bm{D}_X)$ denote the sum-rate distortion function under MML framework with Gaussian mixture observations. In the following we present the main results of this work.
\begin{theorem}\label{Theorem_label_Gaussian_Mixture}
	Without loss of generality, let $H(S|\bm{X})\leq D_S\leq H(\bm{\omega})$ and $0\leq D_{X_i}\leq \bm{e}_i^T\tb{K}_X\bm{e}_i$ for $i=1,2,\cdots,L$, where $\bm{\omega}=(\omega_1,\omega_2,\cdots,\omega_M)$, and $\bm{e}_i$ is an all-zero vector except for 1 at its $i$-th position, then
	\begin{align}
		R^\mathrm{C}(D_S,\bm{D}_X)\geq R^\mathrm{C}_\mathrm{out}(D_S,\bm{D}_X),
	\end{align}
	where
	\begin{subequations}
		\begin{align}
			R^\mathrm{C}_\mathrm{out}(D_S,\bm{D}_X)=&\min_{\bm{\Gamma}\in\mathbb{R}^{L\times L}}H(\bm{\omega})-\beta\left(\bm{\Gamma}\right)+\frac{1}{2}\log_2\frac{\det(\tb{K}_{X})}{\det(\bm{\Gamma})},\label{Cou}\\
			&\text{s.t. }\tb{O}\preceq\bm{\Gamma}\preceq\tb{K}_{X},\label{C3}\\
			&\hspace{0.5cm} \bm{e}_i^T\bm{\Gamma}\bm{e}_i\leq D_{X_i},\text{   for  } i=1,2,\cdots,L.
		\end{align}\label{Cout}
	\end{subequations}
	Herein  
	\begin{align}
		H(\bm{\omega})&=\sum_{\ell=1}^M\omega_\ell\log(\omega_\ell),\\
		\beta\left(\bm{\Gamma}\right)&=\min\left\{D_S,1+\log_2(M-1)p_e\left(\bm{\Gamma}\right)\right\},\label{26b}\\
		p_e\left(\bm{\Gamma}\right) &= \min\left\{1,\min_{0<\alpha<\frac{1}{2}}\left\{2Q\left(\frac{L\alpha}{\sqrt{\mathrm{tr}\{\tb{K}_X\}}}\right)+\frac{1}{L}\left(\frac{1}{2}-\alpha\right)^{-2}\mathrm{tr}\{\bm{\Gamma}\}\right\}\right\}  \label{28}.
	\end{align}
\end{theorem}
\begin{proof}
	See Appendix \ref{proof_Gaussian_Mixture}.
\end{proof}
Eq. \eqref{Cout} is a non-trivial outer bound of sum rate-distortion function of semantic-aware MT problem, which shows the least sum rate of $L$ terminals following Gaussian mixture distribution to meet the fidelity criterion $D_S$ and $\bm{D}_X$. The proof starts from the single letter characterization and two different distortion measures, which yields a nature lower bound of rate. Moreover, we refine the rate of $H(S|\bm{U})$ via Fano inequality, based on the fact that the error probability of detecting $S$ won't be very large if the observable source $\bm{X}$ can be reconstructed with high quality. Note that Eq. \eqref{Cout} still has gap to the true RDF, mainly due to the non-convex optimization in Def. \ref{Definition_Cout} and the unclear characterization of the Gaussian mixture entropy, which are both open questions for now. Nevertheless, the non-convexity optimization can be completely solved by a diagonal assumption on $\bm{\Gamma}$ with sophisticated discussion on distortion regions, which is shown without losing generality in the later. In one word, our outer bound on rate-distortion behavior serves as \textbf{a relative tighter one} for Gaussian mixture sources than Eq. \eqref{29} and Eq. \eqref{24}, even though we set $D_S\rightarrow\infty$, which is mainly attributed to the careful analysis in Eq. \eqref{28}. Next, an toy example is used to illustrate the improvement.
\begin{figure}[t]
	\centering
	\includegraphics[width=0.8\textwidth]{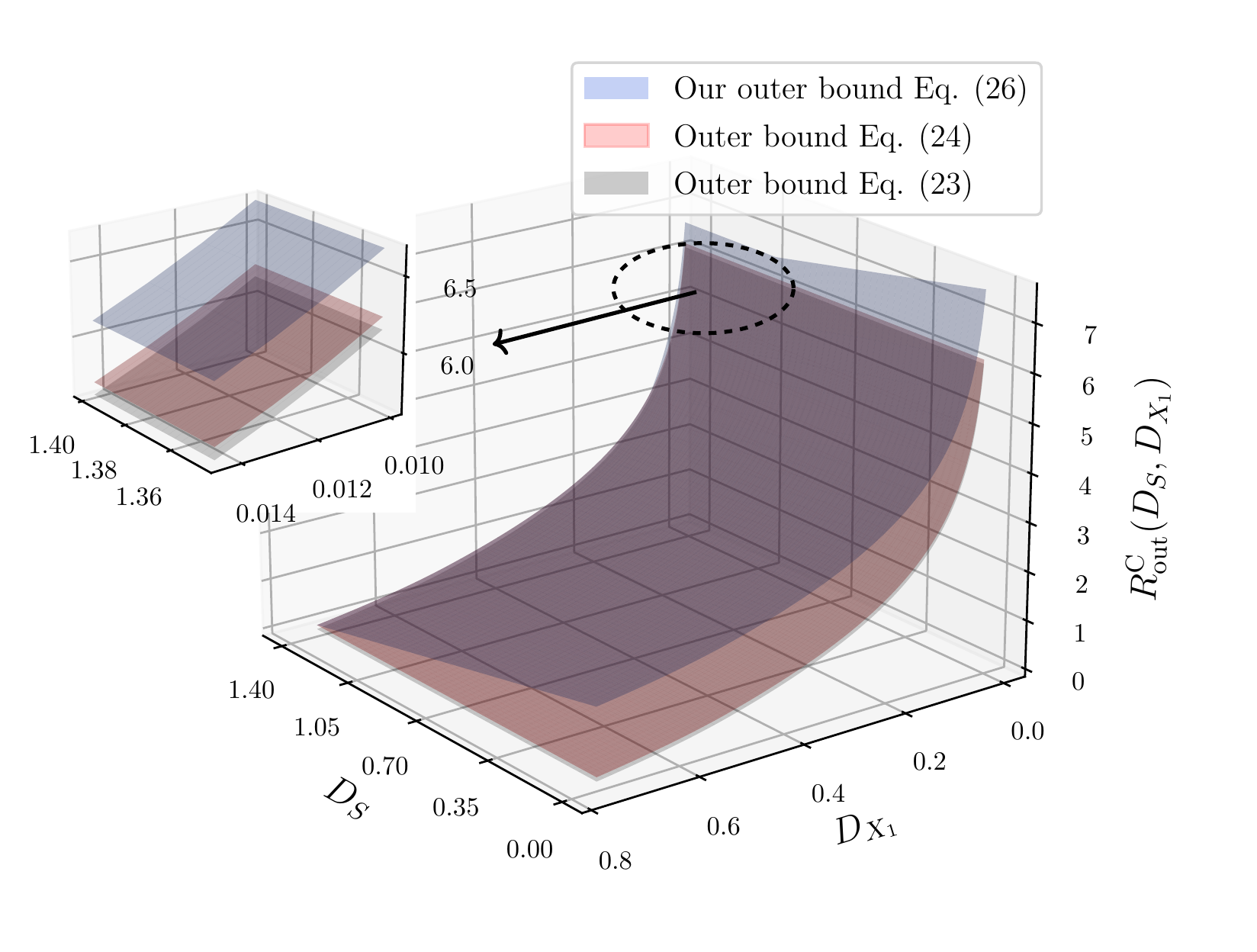}
	\caption{Rate behavior against distortions of the semantic-aware MT source coding problem compared with existing bound}
	\label{3}
\end{figure}

\begin{example}\label{Ex1}
	Let $L=2$, $M=3$, $\bm{\omega}=(0.5,0.2,0.3)$, $\tb{K}_X=\mathrm{diag}\{0.75,0.5\}=\mathrm{diag}\{\sigma_{X_1}^2,\sigma_{X_2}^2\}$, and distortions are constrained as $0\leq D_S\leq H(\bm{\omega})$ and $0\leq D_{X_1}\leq\sigma_{X_1}^2$ and $0\leq D_{X_2}\leq\sigma_{X_2}^2$. Our outer bound $R^\mathrm{C}_\mathrm{out}(D_S,\bm{D}_X)$ is plotted against $D_S$ and $D_{X_1}=D_{X_2}$ in Fig. \ref{3}, shown as the blue surface. Besides, bound $R(\bm{D}_X)=\sum_{i=1}^2R(D_i)$ in Eq. \eqref{29} and Eq. \eqref{24} are also presented as the gray and red surfaces, respectively. 
\end{example}
From Fig. \ref{3}, one can easily conclude the sub-optimality of neglecting the mutual information. Besides, some more observations can be obtained:
\begin{enumerate}
	\item The rate-distortion behavior of $R^\mathrm{C}_\mathrm{out}(D_S,\bm{D}_X)$ shows different curvatures, namely linearity against $D_S$ and logarithm against $\bm{D}_X$. This is owing to the different distortion measures under MML framework. 
	\item Based on the constraint Eq. \eqref{28}, our outer bound is tighter than the trivial bounds at an observation-specified high resolution regime, which means $\mathrm{tr}\{\bm{D}_X\}$ is small enough. The comparison among three bounds is shown in the subplot. It verifies the fact that there exists connection between semantic distortion and observed distortion, In other word, an accurate observation reconstruction will not lead to a terrible semantic recovery, which is not shown in existing works. 
	\item We can explain the above observation in another perspective. It is possible in our scenario that one of the constraint is active while the other is dummy. Therefore, a nature ideas is to explore the activation of the different distortion constraints $D_S$ and $\bm{D}_X$ with our outer bound, which will be shown in the following subsection.
\end{enumerate}
\subsection{Characterization of Outer Bound}
By considering the multi-constraint setups in our semantic-aware rate-distortion problem, it is necessary to investigate the activeness of observation and semantic distortion constraints, thus simplifying the statement of outer bound by figuring out the boundary conditions. Moreover, owing to such an indirect source setting, we call Eq. \eqref{Dx} the surrogate distortion of Eq. \eqref{Ds}, while these two distortions can be transformed to each other with joint Gaussian distributed sources. It is to say, with $(S,\bm{X})$ following joint Gaussian distribution and both under MSE measure, $d_S(S,\hat{S})$ can be converted to $d_X(X_i,\hat{X_i})$ and vice versa. However, in our Gaussian mixture setups, we find the fact that restricting $d_X(X_i,\hat{X_i})$ though reflects an upper bound on $d_S(S,\hat{S})$, but  constraining $d_S(S,\hat{S})$ will not necessarily yield small $d_X(X_i,\hat{X_i})$, for $i=1,2,\cdots,L$.
\begin{figure}[tbp]
	\subfloat[]{
		\includegraphics[width=0.49\textwidth]{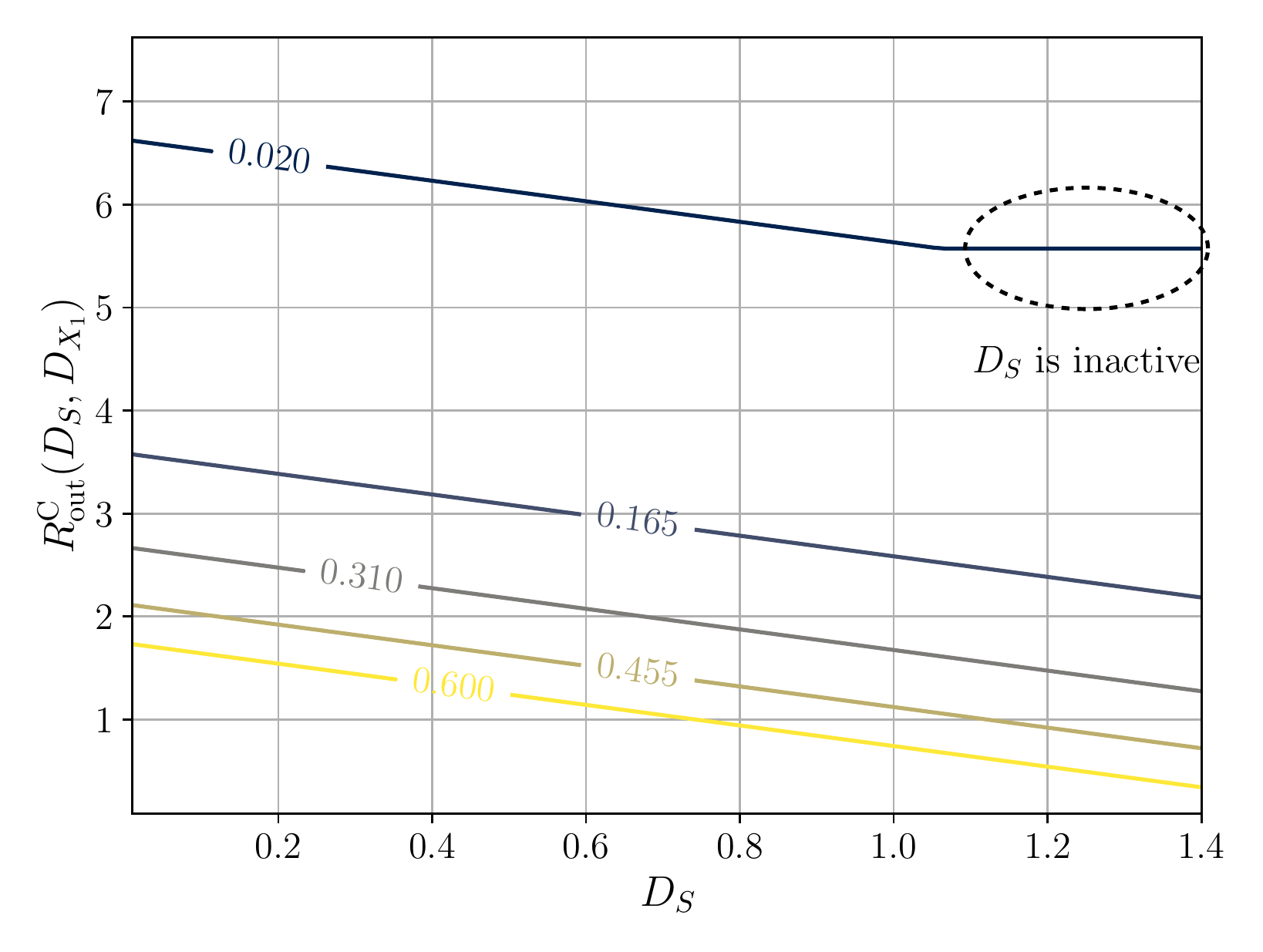}
		\label{}}
	\subfloat[]{
		\includegraphics[width=0.49\textwidth]{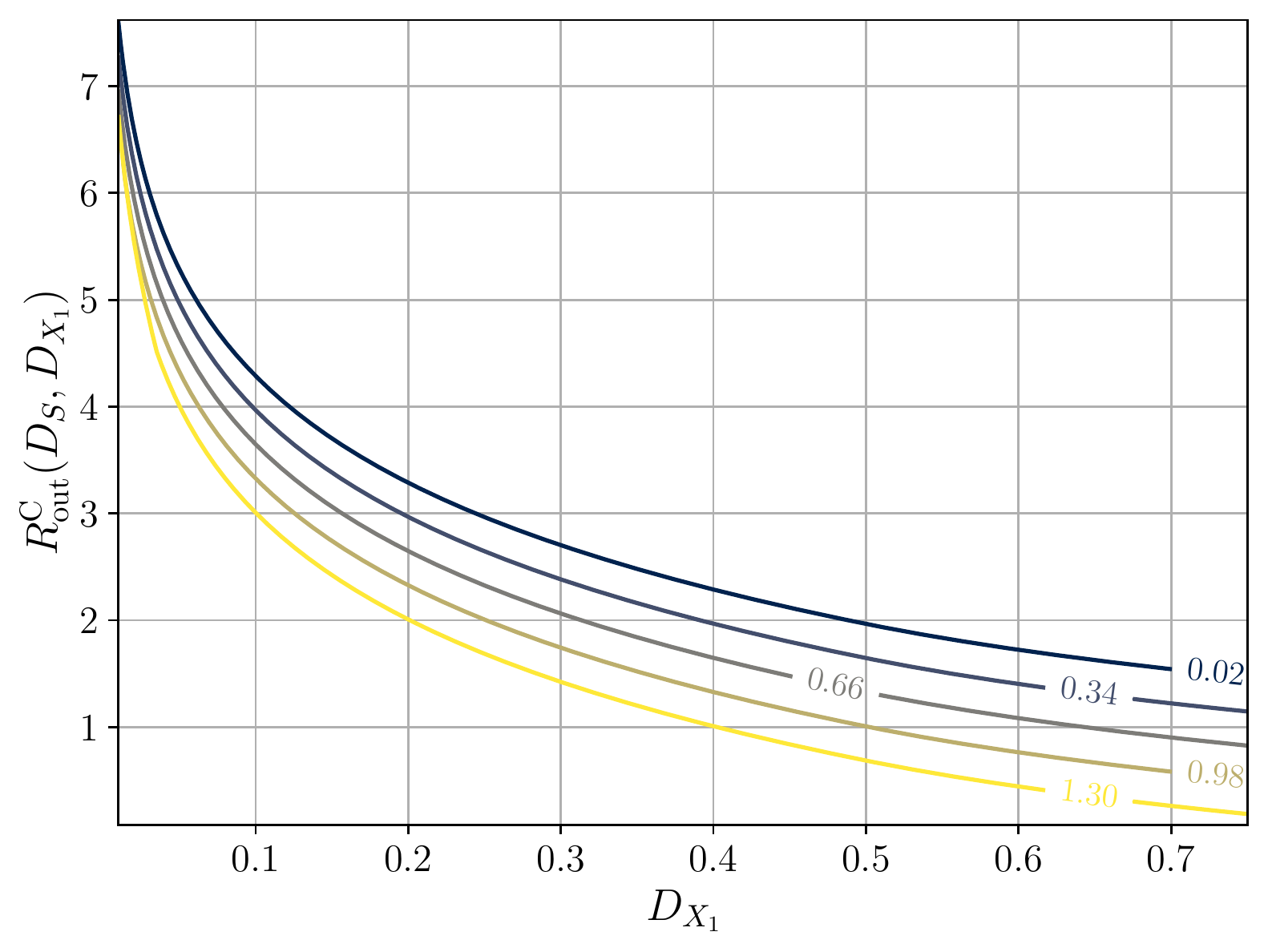}
		\label{}}
	\caption{Rate-distortion behavior: (a) Contour plot of rate against $D_{X_1}$ given $D_S=0.02,0.34,0.66,0.98,1.30$; (b) Contour plot of rate against $D_S$ given $D_{X_1}=0.02,0.165,0.310,0.455,0.600$}\label{Fig3}
\end{figure}

To make it clear, we adopt the setup in Example \ref{Ex1} and present the contour plots of rate outer bound in Fig. \ref{Fig3}. For the left, the rate against semantic distortion constraint is plotted, where we can find the linearity decrease of $D_S$ due to the logarithmic loss. Moreover, for the high resolution regime, e.g. $D_{X_1}=0.02$, semantic distortion constraint is inactive when $D_S\geq1.1$, which is characterized in outer bound Eq. \eqref{Cout}. For the right, an interesting phenomenon in the plot of rate against observation distortion is that the rate is always decrease with $D_{X_1}$ no matter how small $D_S$ is. Therefore we conjugate that in our Gaussian mixture setups, only small $D_{X_1}$ yields small $D_S$, while accurate semantic recovery cannot ensure an accurate observation reconstruction. Before the formal statement, we first assume both semantic and observed distortions lie in the non-trivial intervals, in which we do not consider the case with too large $D_S$ and $\bm{D}_X$ resulting in both inactive constraints, i.e.
\begin{align}
	H(S|\bm{X})\leq &D_S\leq H(\bm{\omega}),\notag\\
	0\leq &D_{X_i}\leq \bm{e}_i^T\tb{K}_X\bm{e}_i \quad\text{for}\quad i=1,2,\cdots,L.
\end{align} 
In the following, a corollary is presented to prove the conjugate theoretically.
\begin{corollary}\label{Contradiction}
	Under the MML framework depicted in Sec. \ref{fundamental_limits}, then for rate-distortion bound $R^C_{\mathrm{out}}(D_S,\bm{D}_X)$,
	\begin{enumerate}
		\item $\bm{D}_X^\star$ exists,
		\item $D_S^\star$ does not exist,
	\end{enumerate}
	where $\bm{D}_X^\star$ and $D_S^\star$ are defined in Def. \ref{dummy}.
\end{corollary}
\begin{proof}
	We proof the Corollary by reduction to absurdity. See Appendix \ref{proof_contradiction}.
\end{proof}
Corollary \ref{Contradiction} reveals the fact that the observation distortion constraint is always active no matter how demanding the semantic constraint is, for Gaussian mixtures. This fact is different to the RDF in which semantic and observed sources are joint Gaussian distributed, which was discussed in \cite[Sec.~\RNum{5}]{Liu_Shao_Zhang_Poor_2022} and \cite[Exp.~2]{Stavrou2022}. From the former \cite[Tab.~1]{Liu_Shao_Zhang_Poor_2022}, one can find the authors divided the activity of constraints into four cases, namely both semantic and observation distortions having choices of activeness and inactiveness. Nevertheless, this corollary prove that there are only two cases for our outer bound, since $D_X$ is always active. This phenomenon  can be attributed to the linearity between semantic and observation sources of Gaussian distributions with MSE loss, yielding a bijective relation of two distortions in that case, while there exists only injection from observation to semantic distortions under discrete-continuous-sources mixed cases. In the sequel, we shall establish a detailed characterization for simplifying statements of our outer bound, according to the above analysis.
\begin{figure}[t]
	\centering
	\includegraphics[width=0.75\textwidth]{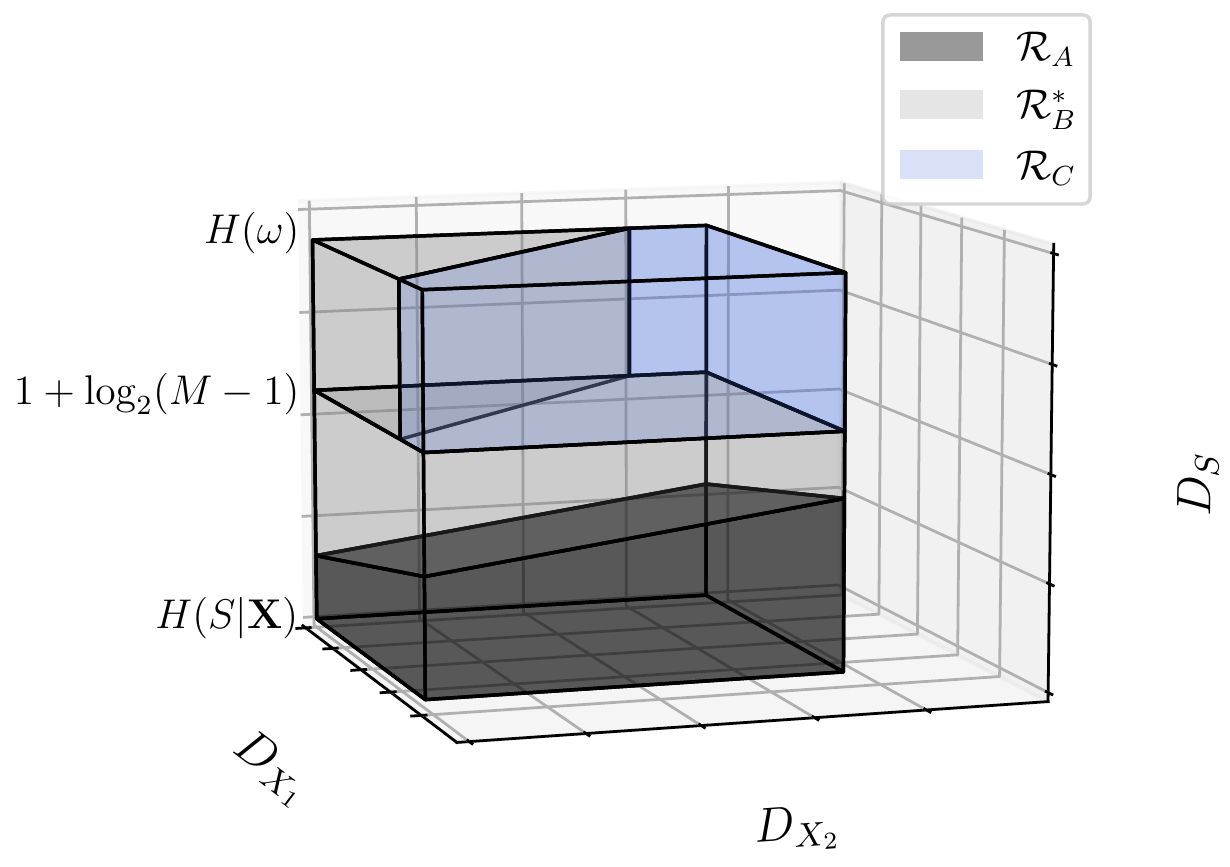}
	\caption{Characterization of distortion regions under the toy example in Exp. \ref{Ex1}.}
	\label{Fig4}
\end{figure}
\begin{corollary}\label{Corollary_region}
	Given the outer bound Eq. \eqref{Cout}, with $H(S|\bm{X})\leq D_S\leq H(\bm{\omega})$ and $0\leq D_{X_i}\leq \bm{e}_i^T\tb{K}_X\bm{e}_i$ for $i=1,2,\cdots,L$, and let
	\begin{align}
		\alpha^*=\mathrm{argmin}_{0<\alpha<\frac{1}{2}}\left\{2Q\left(\frac{L\alpha}{\sqrt{\mathrm{tr}\{\tb{K}_X\}}}\right)+\frac{1}{L}\left(\frac{1}{2}-\alpha\right)^{-2}\mathrm{tr}\{\bm{\Gamma}\}\right\},
	\end{align}
the assumption of $\bm{\Gamma}=\mathrm{diag}\{\gamma_1,\cdots,\gamma_L\}$ will not lose the generality. Moreover, we introduce the following three regions, namely
	\begin{align}
	\mathcal{R}_{A}=\Bigg\{(D_S,\bm{D}_X):D_S\leq 1+\log_2(M-1),D_S\leq1+\log_2(M-1)p_e\left(\sum_{i=1}^LD_{X_i}\right)\Bigg\}.
\end{align}
\begin{align}
	&\mathcal{R}_B^*=\Bigg\{(D_S,\bm{D}_X):D_S\leq 1+\log_2(M-1),D_S\geq1+\log_2(M-1)p_e\left(\sum_{i=1}^LD_{X_i}\right)\Bigg\}\notag\\
	&\hspace{4cm}\bigcup\Bigg\{(D_S,\bm{D}_X):D_S\geq 1+\log_2(M-1),p_e\left(\sum_{i=1}^LD_{X_i}\right)\geq1\Bigg\}.\\
	&\mathcal{R}_C=\Bigg\{(D_S,\bm{D}_X):D_S\geq 1+\log_2(M-1),p_e\left(\sum_{i=1}^LD_{X_i}\right)\geq1\Bigg\},
\end{align}
where 
\begin{align}
	p_e\left(T\right) &= 2Q\left(\frac{L\alpha^*}{\sqrt{\mathrm{tr}\{\tb{K}_X\}}}\right)+\frac{1}{L}\left(\frac{1}{2}-\alpha^*\right)^{-2}T.\label{37}
\end{align}
The characterization of specific statements in terms of different regions can be formulated as
	\begin{enumerate}[a)]
		\item when $(D_S,\bm{D}_X)\in\mathcal{R}_A$,
		\begin{align}
			R^\mathrm{C}_{\mathrm{out}}(D_S,\bm{D}_X)=H(\bm{\omega})-D_S+\frac{1}{2}\log\frac{\det\{\tb{K}_X\}}{\prod_{i=1}^LD_{X_i}}.
		\end{align}
		\item when $(D_S,\bm{D}_X)\in\mathcal{R}_B^*$,
	\begin{align}
		R^\mathrm{C}_{\mathrm{out}}(D_S,\bm{D}_X)=H(\bm{\omega})-\left(1+\log_2(M-1)\right)p_e\left(\sum_{i=1}^LD_{X_i}\right)+\frac{1}{2}\log\frac{\det\{\tb{K}_X\}}{\prod_{i=1}^LD_{X_i}}.
	\end{align}
		\item when $(D_S,\bm{D}_X)\in\mathcal{R}_C$,
	\begin{align}
		R^\mathrm{C}_{\mathrm{out}}(D_S,\bm{D}_X)= H(\bm{\omega})-(1+\log_2(M-1))+\frac{1}{2}\log\det\{\tb{K}_X\}-\frac{L}{2}\log p_e^{-1}\left(1\right).\label{}
	\end{align}
			\end{enumerate}
		\end{corollary}
		\begin{proof}
			See Appendix \ref{proof_corollary_classification}.
		\end{proof}
		Corollary \ref{Corollary_region} provides a full depiction of our outer bound in terms of feasible regions w.r.t. distortions regions. This depiction is based on the assumption that an diagonal matrix $\bm{\Gamma}$ is optimal, which we will prove it does not lose generality later. The proof relies on Corollary \ref{Contradiction} which states that our outer bound can be merely divided into two cases, namely semantic distortion is active or not, in which each case can be divided into high/low resolution observations, respectively. Moreover, to obtain the complete characterization, we further reduce outer bound into standard optimization problems in terms of different distortion regions $(D_S,\bm{D}_X)$ with sophisticated discussions. A simple illustration of feasible regions of the toy example for distortion tuples is provided in Fig. \ref{Fig4}.
		
		As a nature idea, we wish to obtain a tight inner bound to approach outer bound. Unfortunately, we cannot obtain a Berger-Tung inner bound according to Prop. \ref{Theorem_label_inner_bound} in conditional Gaussian semantic-aware MT problem, like its counterpart with joint Gaussian sources. This is owing to the fact that the finite and infinite alphabets of our concerned sources $S$ and $\bm{X}$, which prevents us from constructing an optimal codebook via random binning and hardly conducts the performance analysis. However, motivated by Yang and Xiong \cite{Yang_Xiong_2012}, it is possible for us to design a practical coding scheme via linear block code and quantizers with sub-optimal rate-distortion performance.
		\section{Practical Coding Scheme for MML Framework with Gaussian Mixture Sources}\label{Sec5}
		As mentioned before, Berger-Tung inner bound \cite{Berger1978} based on random binning is not optimal for Gaussian mixture sources, since the codebook cannot compatible with the mixture of continuous and discrete random variables. In this section, we first provide an achievable inner bound based on the idea to decouple the semantic-aware MT problem equipped with Gaussian mixture sources. Then, a corresponding practical coding design is presented to verify its feasibility. Moreover, numerical comparison between inner and outer bounds are considered, in combination with some simulations.
		
		\subsection{An Achievable Inner Bound with Gaussian Mixture Sources}\label{5.1}
		In this section, we provide asymptotic analysis on RDF of an achievable coding design, which relies on the feasible modules instead of theoretical tools like type of method. For simplicity, we make some specific assumption and formulate the inner bound as following:
		\begin{theorem}\label{Theorem_label_coding_scheme}
			Let $R^\star(D_S,\bm{D}_X)$ denote the theoretical sum-rate distortion bound of the achievable coding scheme. Moreover, by fixing $\epsilon>0$, $L=2$, $M=2$, i.e. $\omega_1=\omega_2=\frac{1}{2}$, and given the distribution of $S$ and $\bm{X}$, we have
\begin{subequations}
				\begin{align}
		R^\star(D_S,\bm{D}_X)&=\min_{d_1,d_2} \left\{1+H_2(\mathscr{P}*\mathscr{D})-\sum_{i=1}^2H_2(d_i)
		+\frac{1}{2}\sum_{i=1}^2\log_2\frac{\sigma_{X_i}^2}{D_{X_i}}\right\}+\epsilon,\label{Cinn}\\
		&\hspace{0.8cm}\text{s.t. }\sum_{i=1}^2H_2(p_i*d_i)-H_2(\mathscr{P}*\mathscr{D})\leq D_S,
	\end{align}
\end{subequations}
			where 
			\begin{align}
				\mathscr{P} = p_1*p_2\notag,\quad\mathscr{D} = d_1*d_2,\quad p_i = Q\left(1/\sigma_{X_i}\right) \text{ for } i=1,2.\notag
			\end{align}
		\end{theorem}
		\begin{proof}
			See Appendix \ref{proof_coding_scheme}.
		\end{proof}
		This inner bound is obtained by a "detect and compress" idea, which means the observed Gaussian mixture is first clustered and then compressed. The proof starts from the codebook splitting at each agents, in which the first part of codebook transmits the semantic label while the second part compresses the Gaussian observations. At the receiver, the unique decoder reconstructs the Gaussian signals according to the received labels. The scheme separates the system model into two problems: a "one semantic source $+L$ agents" CEO problem, and $L$ "approximating independent" Gaussian sources compression. Note that, the "approximating independent" means that the correlation of $\bm{X}$ after clustering though still exists, it can be eliminated by applying distinguished sub-codebooks for different Gaussian components, which will be shown in Appendix \ref{proof_coding_scheme}. Moreover, as outer bound, a similar conclusion on the activeness of distortions for inner bound is observed, i.e. for the Eq. \eqref{Cinn}, i.e. observation distortion can not necessarily be limited when semantic recovery is satisfied.
		\begin{corollary}
			Under the framework depicted in Thm. \ref{Theorem_label_coding_scheme}, if $0\leq D_S\leq H(\bm{\omega})$ and $0\leq D_{X_i}\leq \bm{e}_i^T\tb{K}_X\bm{e}_i$ for $i=1,2$, 
			then for rate-distortion bound $R^\star(D_S,\bm{D}_X)$,
			\begin{enumerate}
				\item $\bm{D}_X^\star$ exists,
				\item $D_S^\star$ does not exist.
			\end{enumerate}
		\end{corollary}
		\begin{proof}
			This corollary is a direct conclusion according to definition \ref{dummy} since $R^\star(D_S,\bm{D}_X)$ is monotonically decreasing in terms of $D_{X_1}$ and $D_{X_2}$, yielding no such $D_S< D_S^\star$ satisfying $R^\star(D_S,\bm{D}_X+ \bm{\Delta}_X)=R^\star(D_S,\bm{D}_X).$
		\end{proof}
		\subsection{Corresponding Coding Scheme to Inner Bound}\label{5.2}
		\begin{figure}[htbp]
			\centering
			\includegraphics[width=1\textwidth]{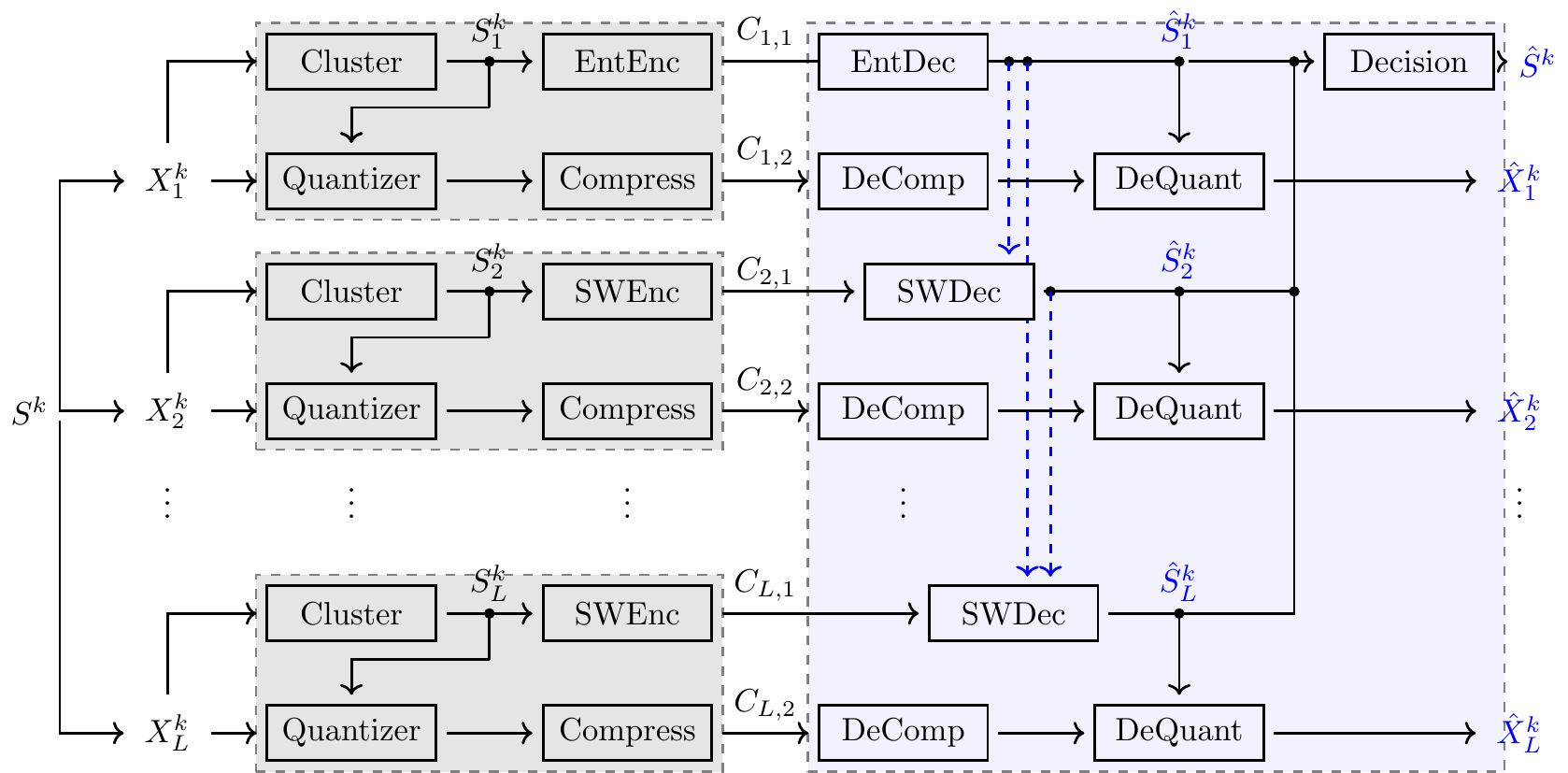}
			\caption{A practical coding design of MML Framework with Conditional Gaussian Distributed Sources with modules: Cluster ($\mathcal{C}$); Quantizer ($\mathcal{Q}$); EntEnc ($\mathcal{T}$): entropy encoder; Compress: lossy encoder ($\mathcal{E}$); SWEnc: Slepian-Wolf encoder ($\mathcal{SW}$); Decomp: lossy decoder ($\mathcal{E}^{-1}$); SWDec: Slepian-Wolf decoder ($\mathcal{SW}^{-1}$); DeQuant: dequantizer ($\mathcal{Q}^{-1}$) and Decision ($\mathcal{D}$).}
			\label{Achiev}
		\end{figure}
		The proposed coding scheme is completely described in Fig. \ref{Achiev}, which consists of modules involving cluster, quantizer, entropy codec and asymmetric Slepian-Wolf codec \cite{Rimoldi_Urbanke_1997,Yang_Stankovic_Xiong_Zhao_2008}. Moreover, we give the detailed realization of each modules in the following.
		
		Specifically, for $i =1,2,\cdots,L$, $\mathcal{C}_i$ first clusters the sampled $k$-length signal $x_i^k$, and output an estimation of the semantic state $S$ independently. Meanwhile the quantizer $\mathcal{Q}_{i}$ quantizes the input $x_i^k$ according to the clustered sequence and outputs integral according to the step $q$, and designs different codebook for different Gaussian components of observations, namely
		\begin{align}
			\mathcal{C}_i(\cdot):&\mathbb{R}^k\mapsto\mathcal{S}^k,\notag\\
			\mathcal{Q}_{i}(\cdot,q):&\mathbb{R}^k\times\mathcal{S}^k\mapsto\{1,2,\cdots,2^{kq}\}.\notag
		\end{align}
		Next, the information to transmit at $i$-th user contains both discrete semantic label $S_i$ and quantized bits of the conditional Gaussian signal. For the semantic labels, we adopt asymmetric Slepian-Wolf (SW) encoder with fixed code rate $R_{i,1}$
		\begin{align}
			\mathcal{T}(\cdot)&:\mathcal{S}^k\mapsto\left\{1,2,\cdots,2^{kR_{1,1}}\right\},\notag\\
			\mathcal{SW}_{i}(\cdot)&:\mathcal{S}^k\mapsto\left\{1,2,\cdots,2^{kR_{i,1}}\right\},\quad\text{where}\quad i\geq2 .\notag
		\end{align}
		Note that given the estimated label, we use lossy encoder to compress $X_i|S_i=s_i$ with rate $R_{i,2}$. Herein a slight difference happens at the first agent, while the first label shares the lossy encoder with the quantized bits.
		\begin{align}
			&\mathcal{E}_{i}(\cdot):\{1,2,\cdots,2^{kq}\}\mapsto\left\{1,2,\cdots,2^{kR_{i,2}}\right\} .\notag
		\end{align}
		The non-interactive encoding processes are illustrated as gray region in Fig. \ref{Achiev} at different users. Overall
		\begin{align}
			C_1 &= [\mathcal{T}(\mathcal{C}_1(X_1)),\mathcal{E}_{1}(\mathcal{Q}_{1}(X_1,q))],\notag\\
			C_i&=[C_{i,1},C_{i,2}] = [\mathcal{SW}_{i}(\mathcal{C}_i(X_i)),\mathcal{E}_{i}(\mathcal{Q}_{i}(X_i,q))],\notag
		\end{align}
		and $R_i = R_{i,1}+R_{i,2}$ for $i=1,2,\cdots,L$. For the unique decoder, lossy decoders are the inverse of encoder, while SW decoders need to collect all previous users' estimation of semantic label to complete the decoding, namely,
		\begin{align}				 
			&\mathcal{T}^{-1}(\cdot):\left\{1,2,\cdots,2^{kR_{1,1}}\right\}\mapsto\mathcal{S}^k,\notag\\
			&\mathcal{SW}^{-1}_{i}(\cdot,\bm{s}^{i-1}):\left\{1,2,\cdots,2^{kR_{i,1}}\right\}\times\mathcal{S}^{k\times(i-1)}\mapsto\hat{\mathcal{S}}^k,\notag\\
			&\hspace{3cm}\text{where}\quad \bm{s}^{i-1}=\{s^k_1,s^k_2,\cdots,s^k_{i-1}\},\notag\\
			&\mathcal{E}^{-1}_{i}(\cdot):\left\{1,2,\cdots,2^{kR_{i,2}}\right\}\mapsto\{1,2,\cdots,2^{kq}\},\notag\text{ for}\quad i\geq2 .\notag
		\end{align}
		The decoded estimated labels are also sent into the dequantizer to reconstruct the observation $\hat{X}_1,\hat{X}_2,\cdots,\hat{X}_L$, and finally the decoder output the ultimate soft estimation of label $\hat{S}$ according to $\mathbb{P}\{\hat{s}^k|\hat{s}_1^k,\hat{s}_2^k,\cdots,\hat{s}_L^k\}$.
		\begin{align}
			\mathcal{Q}^{-1}_{i}(\cdot,q):&\{1,2,\cdots,2^{kq}\}\times\hat{\mathcal{S}}^k\mapsto\mathbb{R}^k,\notag\\
			\mathcal{D}(\cdot):&\hat{\mathcal{S}}^{k\times L}\mapsto[0,1]^k\notag.
		\end{align}
		For practical implementation, we adopt Neyman Pearson (NP) theorem for clustering, dithered trellis coded quantizer (TCQ) for quantization, low density parity check (LDPC) codes for asymmetric SW coding and lossy compression. The details are shown as follows.
		
		\textbf{Neyman Pearson theorem}. NP is a well-known method in hypothesis testing, which utilizes the log likelihood ratio (LLR) of different conditional distribution as decision standard, namely $\frac{P(X=x|S=\ell)}{P(X=x|S\neq\ell)}\geq t$, where $t$ denotes the decision threshold.  
		
		\textbf{Dithered trellis coded quantization}. TCQ is a powerful technique for implementing high-dimension vector quantization, where a continuous signal is assigned to the expanded signal set and its corresponding subsets. Furthermore, with the addition of an i.i.d. uniform dither noise, we can obtain an independent
		quantization sequence which is able to compress the continuous signal nearly lossless at its entropy rate. The quantization error $\epsilon^Q$ can be calculated via the normalized second moments
		of Voronoi region in TCQ \cite{Yang_Xiong_2012} and goes to 0 with small enough distortion. For more details one can turn to \cite[\RNum{4}.~A]{Yang_Xiong_2012}.
		
		\textbf{LDPC code for lossy compression and asymmetric SW coding}. It is not rare to use LDPC channel codes for source coding scheme. For lossy compression,  syndrome based encoding of LDPC code with parity matrix $H_{(n-n_1)\times n}$ are utilized, which has been verified to achieve the rate-distortion behavior of binary sources when equipped with an powerful ML decoder, namely 
		\[
		\mathrm{syn} = xH^T,
		\] and the rate is calculated as $\frac{n-n_1}{n}\geq1$. The reader can turn to \cite{Matsunaga_Yamamoto_2003} for theoretical details on MacKay’s ensemble. For  asymmetric SW, we apply entropy encoder at first user for $s_1$, while conditional encoder for the second user who selects an LDPC code with code rate $H(S_2|S_1)$ \cite{Liveris_Zixiang_Xiong_Georghiades_2002, Stankovic_Liveris_Xiong_Georghiades_2006}. Its encoder sends the syndrome $s_2H^T$ to the decoder, then the belief propagation based decoder recovers $s_2$ with the side information $s_1$. Note that it needs a sign flip on the LLR from check nodes to varaible nodes accoding to the received syndrome, namely
		\begin{align}
			\tanh\left(\frac{\mathrm{LLR}^\mathrm{out}_{\jmath,\imath}}{2}\right)=(-1)^{\mathrm{syn}_\jmath}\prod_{\imath\neq \jmath}\tanh\left(\frac{\mathrm{LLR}^\mathrm{in}_{\jmath,\imath}}{2}\right),
		\end{align}
		
		where $\mathrm{LLR}_{\jmath,\imath}^{\mathrm{in}/\mathrm{out}}$ represents input/output LLR from/to $\jmath$-th check node along $\imath$-th edge, and $\mathrm{syn}_\jmath$ denotes the $\jmath$-th bit of syndrome. The $i$-th user follows the similarly way with a selection of a codebook with rate $\frac{k}{n}=H(S_i|S_1,\cdots,S_{i-1})$, for $i=2,3,\cdots,L$. For the logarithmic distortion in Eq. \eqref{16} with $\hat{s_j}(s_j)=\mathbb{P}\{s_j|s_{1,j},s_{2,j},\cdots,s_{L,j}\}$ where $j=1,2,\cdots,k$, we utilize the empirical distribution
		\begin{align}
			\mathbb{E}\left[d_S(S,\hat{S})\right]&=\frac{1}{k}\sum_{j=1}^k\log\frac{1}{\mathbb{P}\{s_j|s_{1,j},s_{2,j},\cdots,s_{L,j}\}}\notag\\
			&=\sum_{s^k,\bm{s}^L}\mathbb{P}\{s^k,\bm{s}^L\}\log\frac{1}{\mathbb{P}\{s^k|\bm{s}^L\}}\geq H(S|C_\mathcal{L})\notag.
		\end{align}
		\begin{figure}[tbp]
			\centering
			\includegraphics[width=1\textwidth]{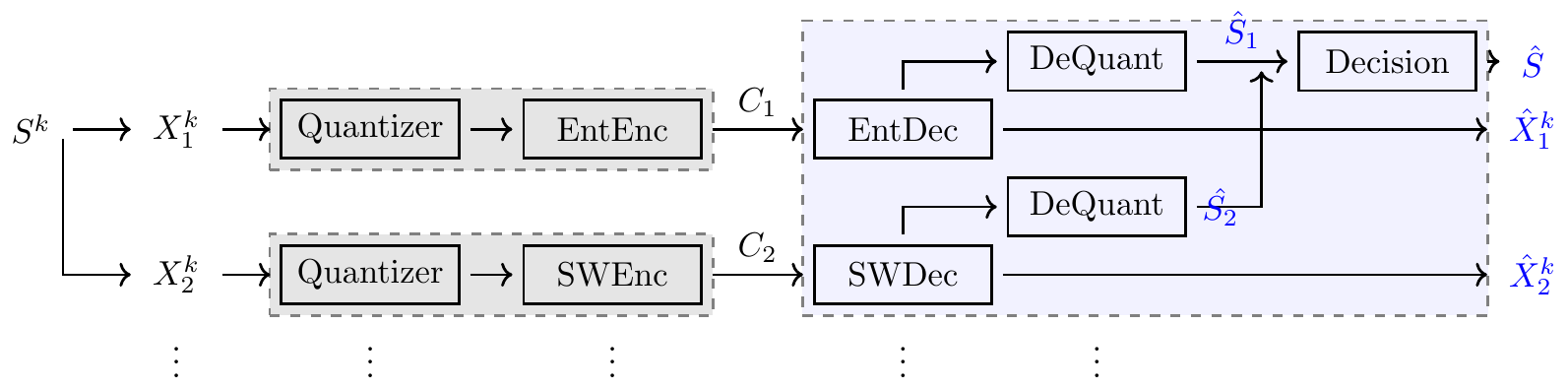}
			\caption{Achievable scheme based on direct compressing (realized in \cite{Yang_Xiong_2012})}
			\label{Achiev1}
		\end{figure}
		One can find that the proposed coding scheme is based on a "detect and compress" idea, and it provides sufficient independence by introducing clustering before quantifying the information among correlated users. For the complexity, the separate design merely applies binning tools (SW coding) on semantic labels instead of the total observations, resulting in a lower cost to find the joint typicality than direct compressing (illustrated in Fig. \ref{Achiev1}). The idea to directly compress and detect is used in \cite{Yang_Stankovic_Xiong_Zhao_2008} for Gaussian sources, which is widely applied in most real world situations. In this section, we emphasize the potential superiority of the "detect and compress" approach over the direct method, as it brings more significant performance gain despite the cost of label clustering. As a preliminary exploration, in the following, we will show that our scheme outperforms the "compress and detect" scheme on the sum rate performance under some setups, e.g. large noises or smaller semantic alphabet.
		
		
		\subsection{Numerical Results}
		In the final part of this section, we provide some simulation results to show the performances.
		
		\textbf{Setups and Baselines}. We assume an MML framework of semantic-aware MT coding problem with conditional Gaussian distributed sources with parameters $L=2$, $M=2$, $\omega_1=\omega_2=0.5$, and $\tb{K}_X=\sigma^2\bm{I}_2$. We choose step of the quantizer $q=3$ bits and dithered noise from $[-0.5,0.5]$ in TCQ (in this symmetric scenario we remove the subscripts since $\sigma^2_{X_1}=\sigma^2_{X_2}$ and $D_{X_1}=D_{X_2}$). We adopt optimal LDPC code with different rates for SW codec, whose distribution degree can be found in \cite[Ch.~7.2]{DBLP:phd/ndltd/Chung00}. The average distortion are computed with logarithmic loss between $S$ and $\hat{S}$, and MSE loss between $\bm{X}$ and $\hat{\bm{X}}$. For theoretical comparison, we plot the converse bound in Eq. \eqref{Cout} and an achievable inner bound in Eq. \eqref{Cinn} with $\epsilon\rightarrow 0$. For practical comparison on simulation results, we adopt the aforementioned scheme in \cite{Yang_Stankovic_Xiong_Zhao_2008} for MT problem.
		\begin{figure}[tbp]
			\subfloat[]{
				\includegraphics[width=0.49\textwidth]{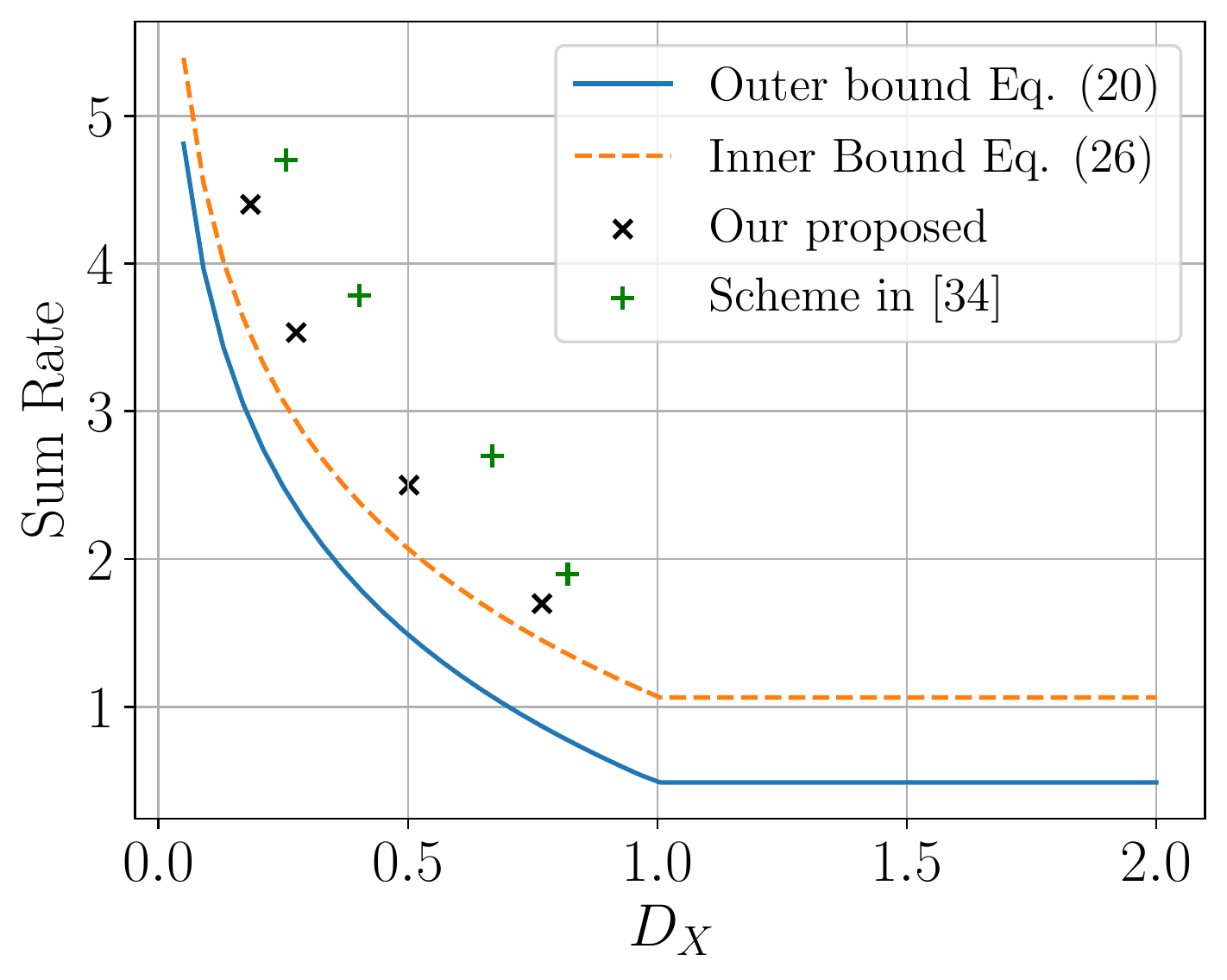}
				\label{4}}
			\subfloat[]{
				\includegraphics[width=0.49\textwidth]{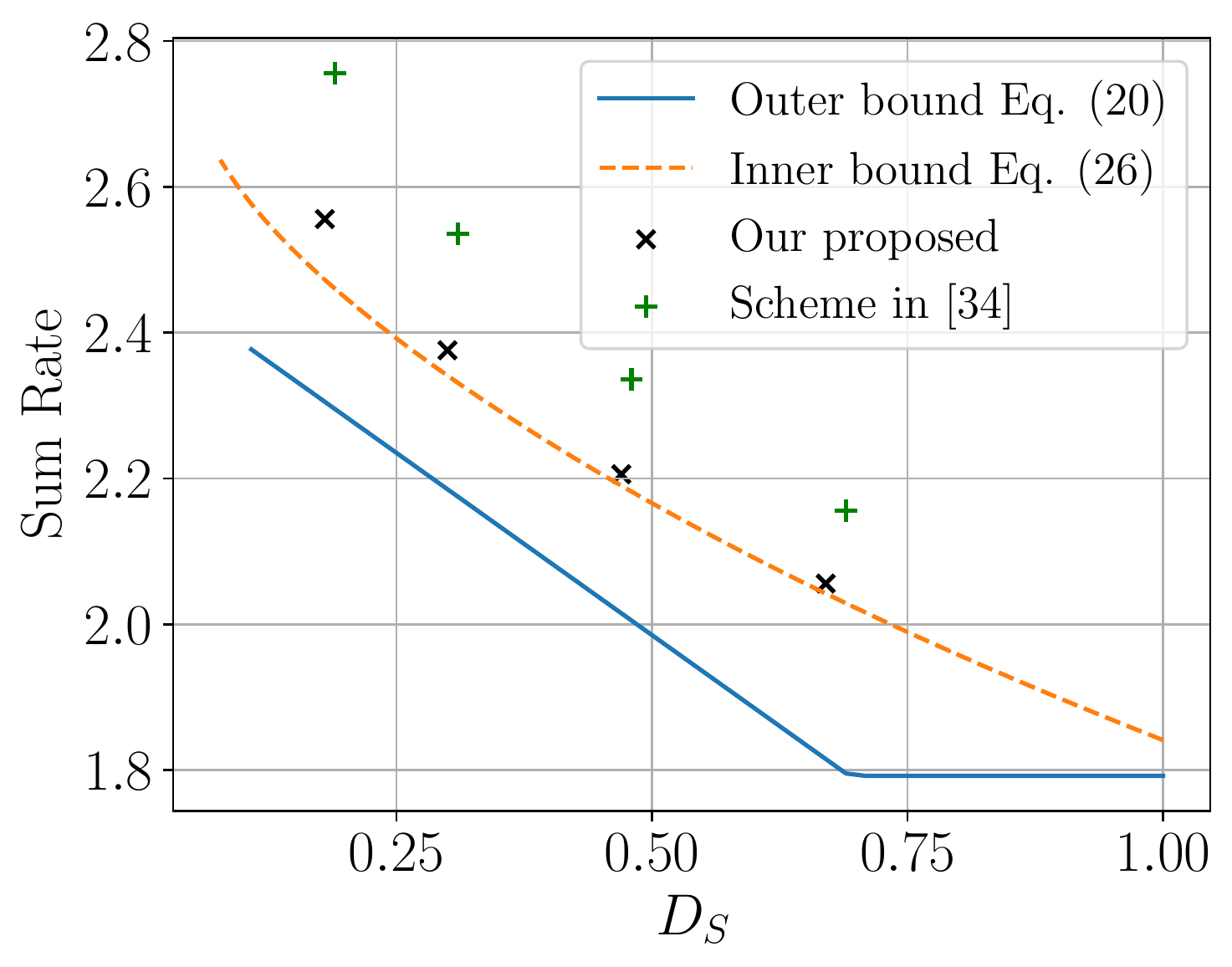}
				\label{5}}
			\caption{Rate-distortion behavior in comparison with scheme in \cite{Yang_Xiong_2012}: (a) against $D_X$ given $D_S=0.05$; (b) against $D_S$ given $D_X=0.2$}\label{Fig5}
		\end{figure}

		\textbf{Rate performance against distortions}. In Fig. \ref{4}, we investigate sum rate against $D_X$ in terms of theoretical bounds (lines) and simulation results (scatters). The sum rate decreases with the required distortion, while the curve shows non-linearity against $D_X$ since the logarithmic behavior of Gaussian entropy. Moreover, the simulation results show the superiority of our scheme than its competitor in \cite{Yang_Stankovic_Xiong_Zhao_2008} in terms of sum rate, though the former consumes extra bits to describe the semantic label. The performance of our scheme not only benefits from the lower complexity on binning two semantic labels $s_1$ and $s_2$, but also take advantages of few quantize error on recovering the true label. These two reasons result in the better performance which is shown in Fig. \ref{4}. In Fig. \ref{5}, the similar plot is presented against $D_S$ in terms of the same arguments. The difference on curvature between theoretical bounds can be originated from the $H(\bm{\omega})-D_S$ in outer Eq. \eqref{Cou} and MT optimization problem (non-linear) in inner bound Eq. \eqref{Cinn}. Besides, it is natural to find the better performance of our scheme, since the detection of $s$ at receiver enlarges the distortion according to noisy signals.
		
		\textbf{Rate allocation between users}. Fig. \ref{66} shows the rate allocation between two users, in terms of the same arguments. It provides another perspective to illustrate the superiority of our scheme. The gap between the simulation points and the inner bound mainly comes from the quantized error and restriction of block length.
		\begin{figure}[tbp]
			\centering
			\includegraphics[width=0.5\textwidth]{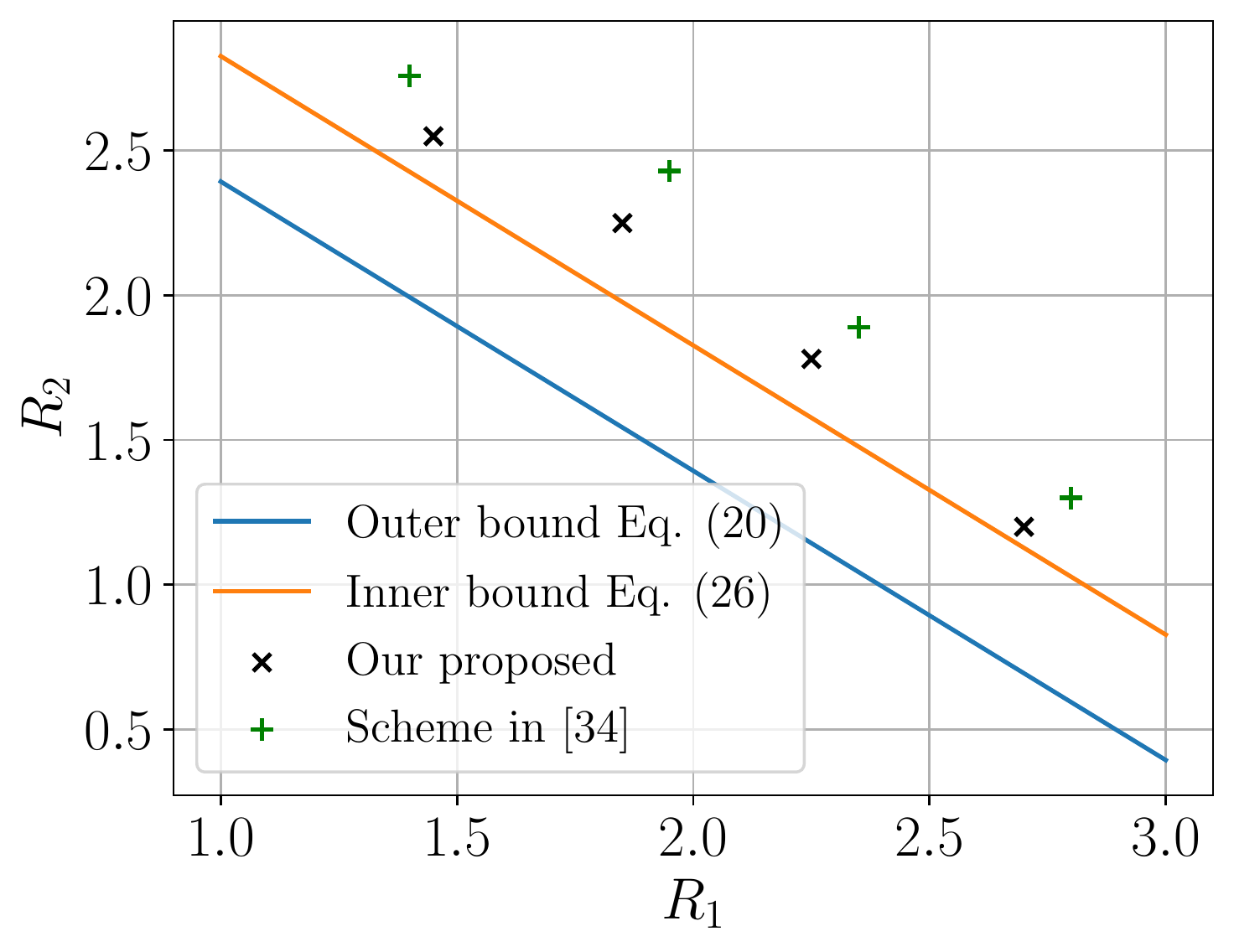}
			\caption{Rate-allocation behavior given $D_S=0.05,D_X=0.2$}
			\label{66}
		\end{figure}
		
		\textbf{Influence of noise} In Fig. \ref{7}, achieved distortions are plotted against SNR of the Gaussian mixture, i.e. $\sigma^{-2}$. We can find both two distortions descends with the decrease of noise. For the semantic distortion, our scheme enjoys an overwhelming advantage than the scheme in \cite{Yang_Stankovic_Xiong_Zhao_2008} owing to the semantic-awareness. For the observed distortion, our scheme outperforms the competitor scheme especially at the worse noise scenarios, since we use extra bits to preserve the labels.
		Besides, in Fig. \ref{88}, we compare inner and outer bounds with the increase of SNR. It can be validated that the gap between two bounds will be reduced with a smaller $\sigma^2$. Intuitively, it means that our our inner bound is able to approach outer bound in spite of the untightness, especially for the cases that the semantic label can be clustered correctly. This fact verifies the potentials of our coding scheme for a semantic-aware compression problem since it is able to approach the theoretical limits under some special cases. 
		\begin{figure}[tbp]
			\subfloat[]{
				\includegraphics[width=0.49\textwidth]{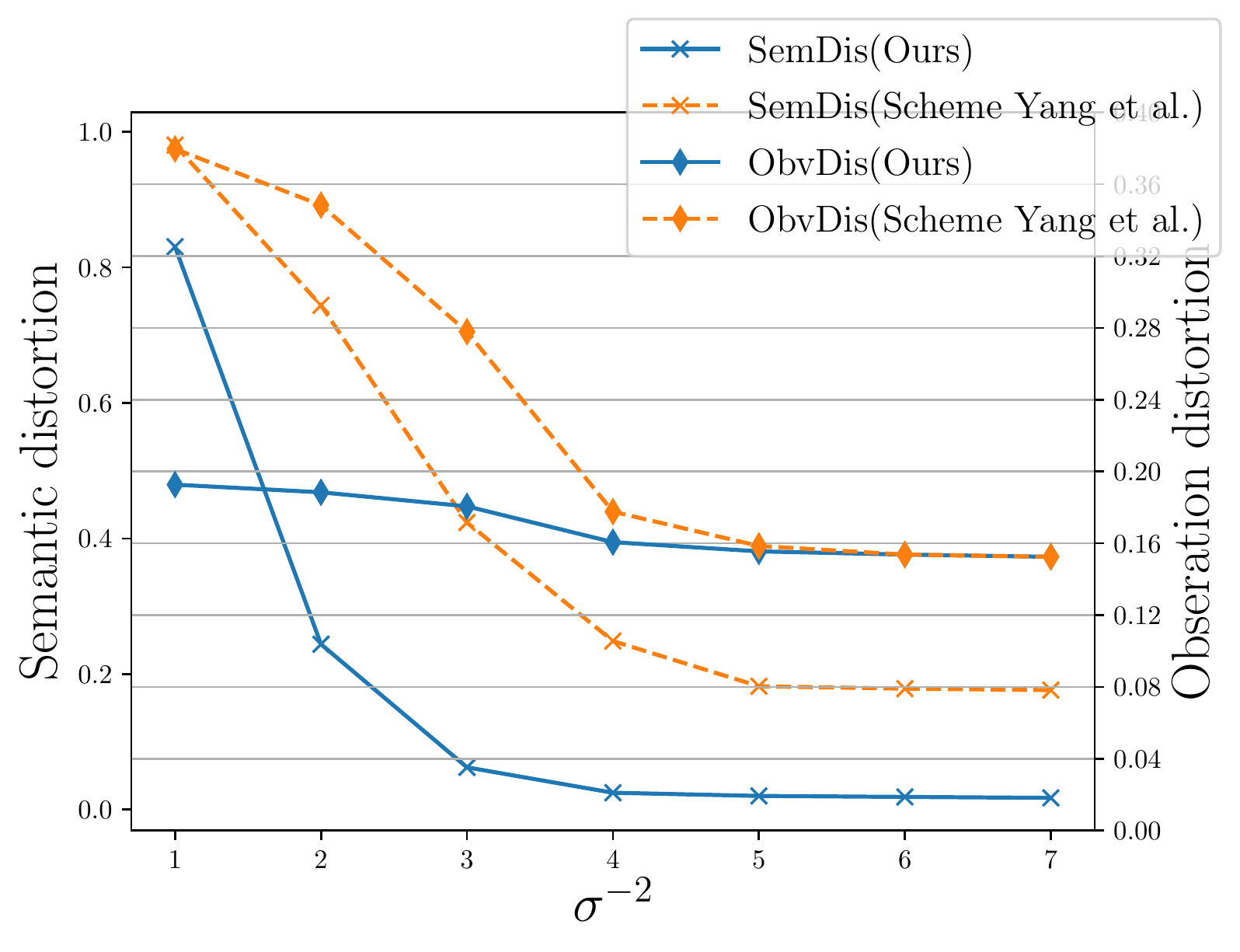}
				\label{7}}
			\subfloat[]{
				\includegraphics[width=0.49\textwidth]{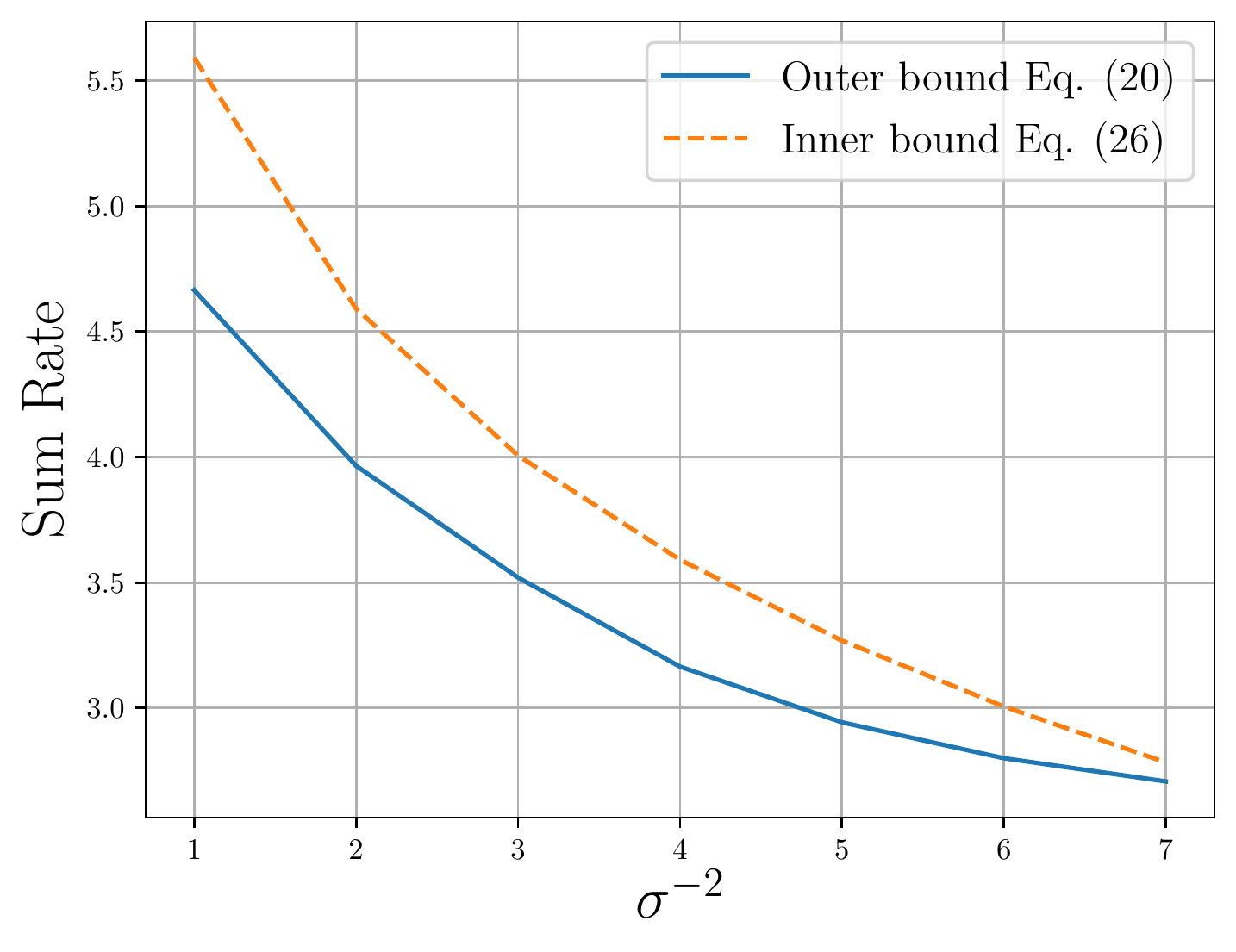}
				\label{88}}
			\caption{(a) Distortion behavior against SNR given $R=3.85$, where SemDis: Semantic distortion and ObvDis: Observation distortion; (b) Sum rate behavior against SNR given $D_S=0.05$, $D_X=0.2$ in terms of our inner and outer bounds.}\label{Fig8}
		\end{figure}
		\section{Conclusion}\label{Sec6}
		In this paper, we investigate the RDF problem of the proposed semantic-aware MT source coding model, which is based on a classic CEO problem and introduces extra observed constraints. We first present the generalized form of outer and inner bounds and discuss the degeneration cases. Moreover, by considering the real world sources, we specify the single letter characterization to Gaussian mixture sources, which models the practical semantics over a finite alphabet and observations conditional Gaussian distributed. The specification incorporating an improved outer bound and an theoretically-analyzable inner bound. We find a fact that, different from the joint Gaussian distributed semantic and observed sources, in our scenario there exists partially trade-off between semantic and observed distortion constraints. In other word, good observation reconstructions means satisfying semantic recovery, but an excellent semantic estimation cannot ensure accurate observation retrievals.
		
		Besides, we also present an practical coding scheme for Gaussian mixture sources based on the inner bound, which is based on "detect and compress" principle. The estimated semantic label and the continuous observations are compressed at the same time, motivating us to explore the potential superiority of the "detect and compress" approach compared to the direct compressing method, since it results in a more substantial performance improvement, even considering the cost of label clustering.
		
		\appendices
		\section{Proof of Theorem \ref{Theorem_label_outer_bound}}\label{proof_single_letter_outer_bound}
		This appendix provides the proof of the generalized outer bound of sum-rate distortion function in the semantic-aware MT source coding system. In this proof, we first introduce a remote source to our system, upon which we obtain conditional independence among observations $\bm{X}$. Next, by verifying the existence of the remote source, we utilize it for argumenting the observed sources and obtain the single letter form of the semantic-aware RDF via chain rule and time sharing. Note that the source argumenting was pioneered by Ozarow \cite{Ozarow}, and it was successively adopted by Wang and Viswanath \cite{Wang_Viswanath}; Wagner and Viswanath \cite{Wagner_Anantharam_2008} for multiple description problem and improved outer bound of MT problem, respectively. Finally, we verify the Markov relation among the distributions of random variables.
		
		For further use, we define $X^k_{i}(q)$ as the observed symbol output at the $i$-th agent at $q$-th time slot/ block component, and abbreviate $X^k_{i}(1:q)\triangleq X^k_i(1)X^k_i(2)\cdots X^k_i(q)$ where $0\leq q\leq k$. We first present the definition of the remote source.
		\begin{definition}\label{Def_Y}
			Let $\mathcal{Y}$ denote the set involving finite-alphabet random variables $\bm{Y}$, in which $\bm{Y}$ satisfies the following constraints:
			\begin{enumerate}[i)]
				\item $\bm{Y}^k$ is a $k$-length vector whose components are mutually independent;
				\item $\bm{Y}^k-\bm{X}^k- C_\mathcal{L}$;
				\item $X_i^k-\bm{Y}^k- X_j^k$ where $1\leq i\neq j\leq L$.
			\end{enumerate}
		\end{definition}
		From the definition, property $i)$ indicates that $\bm{Y}$ is a memoryless source, $ii)$ shows the 'remote' Markov relation and $iii)$ implies the conditional independence. Note that one can readily verify that $\mathcal{Y}$ is nonempty since $\bm{Y}$ can take the same value as $\bm{X}$, namely $\mathcal{Y}$ contains $\bm{Y}=\bm{X}$.	The proof starts from the assumption there exists a remote source $\bm{Y}\in\mathcal{Y}$ such that
		\begin{align}
			k\sum_{i=1}^LR_L\geq& \sum_{i=1}^LH(C_i)\notag\\
			\geq& H(C_{\mathcal{L}})\notag\\
			=& I(\bm{X}^k;C_{\mathcal{L}})\notag\\
			\overset{(a)}{=}& I(\bm{Y}^k,\bm{X}^k;C_{\mathcal{L}})\notag\\
			\overset{(b)}{=}&I(\bm{Y}^k;C_{\mathcal{L}})+\sum_{i=1}^{L}I(X_i^k;C_i|\bm{Y}^k)\notag\\
			\overset{(c)}{=}&\sum_{j=1}^kI(\bm{Y}^k(q);C_\mathcal{L}|\bm{Y}^k(1:q-1))+\sum_{i=1}^L\sum_{j=1}^kI(X^k_{i}(q);C_i|\bm{Y}^k,X^k_{i}(1:q-1))\label{40}
		\end{align}
		where $(a)$ and $(b)$ follow by the ii) and iii) properties of $\bm{Y}$ in Def. \ref{Def_Y}, respectively, and $(c)$ follows from the chain rule of mutual information. Next we lower bound the second term in Eq. \eqref{40}. Note that the fact for $i=1,2,\cdots,L$,
		\begin{align}
			&I(X^k_{i}(q);C_i|\bm{Y}^k)+I(X^k_{i}(q);X^k_{i}(1:q-1)|C_i,\bm{Y}^k)\notag\\
			=&I(X^k_{i}(q);C_i|\bm{Y}^k,X^k_{i}(1:q-1))	+I(X^k_{i}(q);X^k_{i}(1:q-1)|\bm{Y}^k)
		\end{align}
		resulting in 
		\begin{align}
			I(X^k_{i}(q);C_i|\bm{Y}^k,X^k_{i}(1:q-1))\geq I(X^k_{i}(q);C_i|\bm{Y}^k),\label{39}
		\end{align}
		in which the inequality holds due to the i.i.d. source $\bm{X}_i$ thus $I(X^k_{i}(q);X_{i}^k(1:q-1)|\bm{Y}^k)=0$. Hence by substituting \eqref{39} into \eqref{40} we write
		\begin{align}
			\sum_{i=1}^LR_L\geq&\frac{1}{k}\sum_{j=1}^k\Bigg[ I(\bm{Y}^k(q);C_\mathcal{L}|\bm{Y}^k(1:q-1))+\sum_{i=1}^LI(X^k_{i}(q);C_i|\bm{Y}^k)\Bigg]\notag\\
			=&\frac{1}{k}\sum_{j=1}^k\Bigg[ H(\bm{Y}^k(q))-H(\bm{Y}^k(q)|C_\mathcal{L},\bm{Y}^k(1:q-1))\notag\\
			&\hspace{4cm}+\sum_{i=1}^L\left(H(X^k_{i}(q)|\bm{Y}^k)-H(X^k_{i}(q)|C_i,\bm{Y}^k)\right)\Bigg]\label{400}
		\end{align}
		where the equation Eq. \eqref{400} holds since $\bm{Y}$ is a memoryless source. Furthermore, by defining $U_i(q)\triangleq(C_i,\bm{Y}(1:q-1))$ for $i=1,2,\cdots,L$, and a time sharing random variable $Q$ which is uniformly distributed over $\{1,2,\cdots,k\}$ thus independent with source $\bm{Y},\bm{X}$, we obtain
		\begin{align}
			&\frac{1}{k}\sum_{j=1}^k\left[I\left(\bm{Y}^k(q);\bm{U}(q)\right)+\sum_{i=1}^LI\left(X^k_{i}(q);U_i(q)|\bm{Y}^k(q),\bm{Y}(q^c)\right)\right]\notag\\
			\overset{(d)}{=}& I\left(\bm{Y}^k(Q);\bm{U}|Q\right)+\sum_{i=1}^LI\left(X^k_{i}(Q);U_i|\bm{Y}^k(Q),\bm{Y}^k(Q^c),Q\right)\notag\\
			\overset{(e)}{=}&I\left(\bm{Y};\bm{U}|Q\right)+\sum_{i=1}^LI\left(X_{i};U_i|\bm{Y},\bm{W},Q\right),\label{41}
		\end{align}
		where $(d)$ follows the introduction of a new random variable $\bm{U}$ with distribution $P_{\bm{U}(q)|\bm{Y}^k(q)=\bm{y}}=P_{\bm{U}|\bm{Y}=\bm{y},Q=q}$; $(e)$ follows by the random variable $\bm{W}\triangleq\bm{Y}^k(Q^c)$. Besides, the cardinality of $\mathcal{U}_i$ can be bounded by $|\mathcal{X}_i|+2^L+L-2$ \cite[Lemma~7]{Wagner_Anantharam_2008} for all $i=1,2,\cdots,L$. After that, it remains to show that $(R,D_S,\bm{D}_X)$ is admissible. Note the constraints in Eq. \eqref{1} and Eq. \eqref{2} indicate that
		\begin{align}
			&D_S\geq\mathbb{E}d^k_S\left(S^{k}(Q), \hat{S}^{k}(Q)\right) , \notag\\
			&D_{X_i}\geq\mathbb{E}d^k_X\left(X_{i}^{k}(Q), \hat{X}^{k}_{i}(Q)\right) ,\qquad i=1,\cdots,L.\notag
		\end{align}
		for uniformly distributed $Q$. The combination of Eq. \eqref{41} with the above $L+1$ constraints build the single-letter characterization of our outer bound presented in Def. \ref{Definition_Cout}. Finally, to show the joint distribution of random variables, we recall the property iii) of the remote source $\bm{Y}$, thus for arbitrary $i=1,2,\cdots,L$,
		\begin{align}
			&X^k_i(Q^c)-\bm{Y}^k(Q^c)-\bm{X}^k_{i^c}(Q^c)\notag\\
			\overset{(f)}{\Rightarrow}&X^k_i(Q^c)-(X^k_i(Q),\bm{Y}^k(Q^c))-(\bm{X}^k_{i^c})\notag\\
			\overset{(g)}{\Rightarrow}&X^k_i(Q^c)-(X^k_i(Q),\bm{Y}^k(Q^c))-(S^k(Q),\bm{X}^k_{i^c}).\label{42}
		\end{align}
		where $(f)$ is due to the i.i.d source $X_i$ and $(g)$ follows that $I(\bm{X}^k_i(Q^c);S^k(Q)|\bm{X}^k_i)=0$. Moreover, by using the notation $\bm{W}$ since $Q$ is independent with $\bm{X}^k$ and $\bm{Y}^k$, the Markov chain Eq. \eqref{42} can be further written as
		\begin{align}		
			&X^k_i(Q^c)-(X^k_i(Q),\bm{W},Q)-(S^k(Q),\bm{X}^k_{i^c})\notag\\
			\overset{(h)}{\Rightarrow}&U_i-(X^k_i(Q),\bm{W},Q)-(S^k(Q),\bm{X}^k_{i^c})\notag\\
			\overset{(i)}{\Rightarrow}&U_i-(X^k_i(Q),\bm{W},Q)-(S^k(Q),\bm{X}^k_{i^c},\bm{U}_{i^c}),\label{43}
		\end{align} 
		where $(h)$ is due to $U_i$ is function of $(X^k_i(Q),\bm{W},Q)$, and similarly $(i)$ follows that $\bm{U}_{i^c}$ is the function of $(\bm{X}^k_{i^c},\bm{W},Q)$. Eventually, we can remove the time sharing random variable $Q$ since we focus on the sum-rate function in this paper. The combination of Eq. \eqref{43} and Def. \ref{Def_Y} hence motivates the joint distribution of random variables presented in Def. \ref{Definition_Cout}, which completes the proof.
		
		\section{Proof of Corollary \ref{corollary_degeneration}}\label{proof_corollary_degeneration}
		In this appendix, we present sketch proofs of the degenerations of our bounds to existing works in Corollary \ref{corollary_degeneration}.
		\begin{enumerate}
			\item For the single user case, we divide the proof into two steps: we first prove our outer bound and inner bounds coincides to each other, and then prove it is equivalent to the indirect rate-distortion function characterizing semantic information. 
			To avoid confusing, we use the notation $D_X,Y,W,U$ and $V$ instead of their vector forms in the following, respectively, due to the unique agent. In the case where $L=1$, we first show
			\begin{align}
				R_\mathrm{out}(D_S,D_X)\geq R_\mathrm{in}(D_S,D_X).\label{51}
			\end{align}
			Specifically, one can verify this by standardizing the objective functions of both outer and inner bounds into a unify format and subsequently evaluating the feasible regions. To see this, we write
			\begin{align}
				\max_{Y,W}\min_{U} \left\{I(Y;U)+I(X;U|Y,W)\right\}			\overset{(j)}{\geq}&\max_{W}\min_{U} \left\{I(S;U)+I(X;U|S,W)\right\},\notag\\
				\overset{(k)}{\geq}&\min_{U} \left\{I(S;U)+I(X;U|S)\right\},\notag\\
				\overset{(l)}{=}&\min_{U} I(X;U),
			\end{align}
			where $(j)$ follows the fact $S\in\mathcal{Y}$, $(k)$ follows a fixed $W$ and $(l)$ follows the Markov chain $S-X-U$. Therefore, we can obtain
			\begin{align}
				R_\mathrm{out}(D_S,D_X)\geq \widetilde{R}(D_S,D_X)\triangleq&\min_{U} I(X;U)\label{45}\\
				&\text{s.t.}\hspace{0.1cm}\mathbb{E}d_S\left(S;f(U)\right)\leq D_S,\notag\\ 
				&\hspace{0.6cm}\mathbb{E}d_X\left(X_i;g_i(U)\right)\leq D_{X_i}, i=1,\cdots,L.\notag
			\end{align}
			Now it is left to compare the feasible regions. One can readily verify that the joint distribution of random variables in Eq. \eqref{45} is reduced to $P(s)P(x|s)P(u|x)$, which is exactly the same as Eq. \eqref{166} when $L=1$, which means $\widetilde{R}(D_S,D_X)=R_\mathrm{in}(D_S,D_X)$. Consequently, we obtain Eq. \eqref{51} and conclude that $R_\mathrm{out}(D_S,D_X)=R(D_S,D_X)= R_\mathrm{in}(D_S,D_X)$, yielding the tightness between our inner and outer bounds for single user case. Next we try to show
			\begin{align}
				R(D_S,D_X)= \hat{R}(D_S,D_X).\label{56}
			\end{align}
			This is a direct consequence from converting the semantic rate-distortion problem into an equivalent classic rate-distortion problem using the surrogate distortion argument \cite{Kostina2016} \cite[App.~\RNum{1}]{Liu_Shao_Zhang_Poor_2022}, i.e. from $\mathbb{E}[d(S,\hat{S})]\leq D_S$ to $\mathbb{E}[\bar{d}(X,\hat{S})]\leq D_S$. Then the converse part follows
			\begin{align}
				kR(D_S,D_X)\geq H(f(X))=I(X^k;f(X))=I(X^k;U^k)\geq kI(X;U)=kI(X;\hat{S},\hat{X}),
			\end{align}
			which yields $R(D_S,D_X)\geq \hat{R}(D_S,D_X)$. Meanwhile, the achievability parts follows the standard random binning scheme for RDF, thus constructing such a codebook enabling $R(D_S,D_X)\leq \hat{R}(D_S,D_X)$. Finally, the continued equality Eq. \eqref{corollaryp2p} holds by combining Eq. \eqref{51} and Eq. \eqref{56}. 
			\item Under the case when $L=2$, $\bm{D}_X\rightarrow \infty$ and with logarithmic loss, we first show that our outer bound coincides the inner bound, and then prove it is equivalent to the results from Courtade and Weissman \cite{Courtade_Weissman_2014}. Note that here we substitute $D_S$ with $D$ since there lefts the unique distortion constraint. We first show
			\begin{align}
				R_{\mathrm{out}}(D,+\infty)\geq R_{\mathrm{in}}(D,+\infty)\label{633}
			\end{align}
			by the similar method in 1). Therefore we can write 
			\begin{align}
				R_\mathrm{out}(D,+\infty)\geq \bar{R}(D)\triangleq&\min_{\bm{U}} I(\bm{X};\bm{U})\label{500}\\
				&\text{s.t.}\hspace{0.1cm}\mathbb{E}d_S\left(S;f(\bm{U})\right)\leq D_S.\notag
			\end{align}
			Thus it is left to discuss the alphabets of auxiliary random variables. Note that the cardinality bounds on $\mathcal{U}$ can be imposed to $\mathcal{V}$ in Def. \ref{Definition_Cinn} without loss of generality (See \cite[App.~A]{Courtade_Weissman_2014}), meanwhile the joint distribution in Eq. \eqref{500} is reduced to $P(s)\prod_{i=1}^2P(x_i|s)P(u_i|x_i)$ which is exactly Eq. \eqref{166} for the case $L=2$. Therefore we obtain $\bar{R}(D)= R_{\mathrm{in}}(D,+\infty)$, and Eq. \eqref{633} holds. Consequently $R_{\mathrm{out}}(D,+\infty)=R(D,+\infty) =R_{\mathrm{in}}(D,+\infty)$.
			
			Next, we try to show our bounds are equivalent to the bounds of two-user CEO problem with logarithmic loss. For the converse part, with $\bm{U}=(U_1,U_2)$, we write
			\begin{align}
				\min_{\bm{U}}I(S;\bm{U})+\sum_{i=1}^2I(X_i;U_i|S)
				= &\min_{\bm{U}}H(S)-H(S|\bm{U})+\sum_{i=1}^2I(X_i;U_i|S)\notag\\
				\overset{(m)}{\geq}& \min_{\bm{U}}H(S)-\mathbb{E}\left[d(S,\hat{S})|\bm{U}\right]+\sum_{i=1}^2I(X_i;U_i|S)\notag\\
				\geq& \min_{\bm{U}}\left[H(S)-D+\sum_{i=1}^2I(X_i;U_i|S)\right]^+
			\end{align}
			where $(m)$ holds since $H(S|\bm{U}=\bm{u})\leq\mathbb{E}[d(S,\hat{S})|\bm{U}=\bm{u}]\leq D$ when $d(\cdot,\cdot)$ is the logarithmic loss. This implies that $R_{\mathrm{out}}(D,+\infty)\geq R_{\mathrm{out}}^\mathrm{CEO}(D)$. Moreover, for the achievability, one can easily obtain the conclusion with fixed $D$
			\begin{align}
				R_{\mathrm{in}}(D,+\infty)\leq R^\mathrm{CEO}_{\mathrm{in}}(D)\notag
			\end{align} 
			via constructing a reproduction function $f(U_1,U_2,Q)\triangleq\mathbb{P}\{S=s|U_1,U_2,Q\}$ for all $s$. Note that the conclusion that bounds of CEO problem are tight in general, we conclude that 
			\begin{align}
				R_{\mathrm{out}}(D,+\infty)\geq R^\mathrm{CEO}_{\mathrm{out}}(D)=R^\mathrm{CEO}_{\mathrm{in}}(D)\geq R_{\mathrm{in}}(D,+\infty)\label{699}.
			\end{align} Finally, we combine Eq. \eqref{633} and Eq. \eqref{699} to obtain the continued equality Eq. \eqref{corollaryCEO}. 
			\item In this part we wish to prove our outer bound is an improved version of Berger-Tung outer bound, when $D_S\rightarrow\infty$. We can formulate the bound as 
			\begin{align}
				R_\mathrm{out}(\infty,\bm{D})=&\max_{\bm{Y},\bm{W}}\min_{\bm{U}} \left\{I(\bm{Y};\bm{U})+\sum_{i=1}^LI(X_i;U_i|\bm{W},\bm{Y})\right\},\notag\\
				&\text{s.t.}\hspace{0.2cm}\mathbb{E}d_X\left(X_i;g_i(\bm{U})\right)\leq D_{X_i},\hspace{0.1cm} i=1,\cdots,L,\notag
			\end{align}
			for a joint distribution $P_{\bm{W}\bm{X}\bm{Y}\bm{U}}$ of the form
			\begin{align}
				&P(\bm{y}|\bm{x})P(w)\prod_{i=1}^LP(x_i)P(u_i|x_i,w)\label{72}
			\end{align}
			which is exactly the improved outer bound of MT problem from Wagner. Intuitively, a longer Markov chain can be observed in Eq. \eqref{72} than the short Markov chain $U_i\rightarrow X_i\rightarrow \bm{X}_{L/\{i\}}$ in Berger-Tung outer bound, yielding the result $R_\mathrm{out}(\infty,\bm{D})\geq R_{\mathrm{out}}^{\mathrm{BT}}(\bm{D})$. The reader can turn to \cite{Wagner_Anantharam_2008} for more details.
		\end{enumerate}
		\section{Proof of Theorem \ref{Theorem_label_Gaussian_Mixture}}\label{proof_Gaussian_Mixture}
		In general, we aim to find a relaxed lower bound for the single letter outer bound in Def. \ref{Definition_Cout}, since the optimization problem in Def. \ref{Definition_Cinn} is nonconvex for mixture Gaussian sources, and the exact characterization of entropy of mixture Gaussian is still an open question. Based on the fact, we present the proof as follows:
		
		We first obtain a naive outer bound by setting the remote source $\bm{Y}$ as $S$ to confront the non-convexity, due to the aforementioned fact $S\in\mathcal{Y}$. The primary bound relates the conditional entropy $H(S|\bm{U})$ and $H(\bm{X}|S,\bm{U})$ to distortions $D_S$ and $\bm{D}_X$, respectively. Furthermore, a core idea to improve the bound is that the optimal estimate of semantic label $S$ via $\bm{U}$ always outperforms that via $\hat{\bm{X}}$, in which the later can be bounded by concentration inequality and hypothesis testing argument.
		
		Note that the observations follows a conditional vector Gaussian distribution as
		\begin{align}
			p_{\bm{X}}(\bm{x}|S=\ell) \sim \mathcal{N}(\bm{x};\ell\cdot\bm{1},\tb{K}_X).\notag
		\end{align}
		Now equipped with the auxiliary random variables $\bm{U}$, $\bm{\Gamma}\triangleq \tb{K}_{\bm{X}|\bm{U}}$ and $\beta\triangleq H(S|\bm{U})$, we obtain an outer bound for the RDF by revealing the fact that the conditional entropy $\beta$ can be bounded with a function of the correlation matrix $\bm{\Gamma}$, we start with
		\begin{align}
			R &= \min_{\bm{U}}I(\bm{X};\bm{U})\nonumber\\
			&=\min_{\bm{U}}I(S,\bm{X};\bm{U})\nonumber\\
			&=\min_{\bm{U}}I(S;\bm{U})+I(\bm{X};\bm{U}|S)\nonumber\\
			&=\min_{\bm{U}}H(S)-H(S|\bm{U})+h(\bm{X}|S)-h(\bm{X}|\bm{U},S)\nonumber\\
			&\overset{(n)}{\geq}\min_{\beta,\bm{\Gamma}}H(S)-\beta+h(\bm{X}|S)-\frac{1}{2}\log_2\left(2\pi e\right)^L\det(\bm{\Gamma}),\label{beta_single user} 
		\end{align}
		Herein the step $(n)$ holds since 
		\begin{align}
			h(\bm{X}|\bm{U},S)&\leq\frac{1}{2}\log_2\left(2\pi e\right)^L\det\left(\tb{K}_{\bm{X}|\bm{U},S}\right)\notag\\
			&\leq\frac{1}{2}\log_2\left(2\pi e\right)^L\det\left(\tb{K}_{\bm{X}|\bm{U}}\right),
		\end{align}
		where the Gaussian distribution maximizes entropy given second order moment, and $\tb{K}_{\bm{X}|\bm{U},S}\preceq\tb{K}_{\bm{X}|\bm{U}}$ holds trivially. 
		
		In order to lower bound the rate $R$, we want to lower bound each term in Eq. \eqref{beta_single user} respectively. Some terms can be easily obtained straightforwardly. 
		For example, recalling the fact $\sum_{\ell=1}^M\omega_\ell=1$ and all Gaussian components of $\bm{X}$ share a same covariance matrix $\tb{K}_X$, we have
		\begin{align}
			h(\bm{X}|S) &= \sum_{\ell=1}^M\omega_\ell h(\bm{X}|S=\ell)\notag\\
			& = \frac{1}{2}\log_2(2\pi e)^L\det\left(\tb{K}_X\right).\label{111}
		\end{align}
		Moreover, $H(S)$ is a constant when the distribution of semantic source $S$ is given, and we also can bound the main diagonal element of $\bm{\Gamma}$ as
		\begin{align}
			\bm{e}_i^T\bm{\Gamma}\bm{e}_i= \bm{e}_i^T\tb{K}_{\bm{X}|\bm{U}}\bm{e}_i\leq D_{X_i}\quad\text{for}\quad i=1,2,\cdots,L.\label{63}
		\end{align}
		
		Now, the key problem of lower bounding rate $R$ is to upper bound the conditional entropy rate $H(S|\bm{U})$. In the following part, we will establish two different upper bounds respectively. First, an easy upper bound can be found as
		\begin{align}
			H(S|\bm{U})
			=\mathbb{E}_{\bm{U}}[H(S|\bm{U}=\bm{u})]
			\leq\mathbb{E}_{S,\bm{U}}[d_S(S,f(\bm{U}))|\bm{U}=\bm{u}]\leq D_S\label{62},
		\end{align}
		according to the property of logarithmic loss measure. Particularly, the semantic decoder outputs likelihood $f(\bm{U})\in\mathcal{P}_{\hat{S}}$ of $S$ based on $\bm{U}$.
		
		Second, we want to establish a connection between the entropy $\beta$ and matrix $\bm{\Gamma}$. This connection can be interpreted that the semantic distortion will be bounded when fixing the mean square error matrix between observations and its reconstructions. The Fano's inequality can upper bound the conditional entropy $H(S|\bm{U})$ by the error probability of decoding semantic information $S$ as following, 
		\begin{align}
			H(S|\bm{U})&\leq 1+\log_2(M-1)\mathbb{P}\left\{S\neq f_S(\bm{U})\right\}.\label{Fano1}
		\end{align}
		Notice that the decoded semantic information under logarithmic loss $\hat{S}=f(\bm{U})$ is a distribution, hence we concatenate a hard decision with $f(\bm{U})$, which is overall denoted by $f_S(\bm{U})$ on the RHS of Eq. \eqref{Fano1}, where the composite decoder $f_S(\cdot):\prod_{i=1}^L\mathcal{U}_i\mapsto\mathcal{S}$.
		
		As a naturally idea, we can see that the error probability of detecting $S$ won't be very large if the observable source $\bm{X}$ can be reconstructed with high quality. The following steps will show this insight rigorously. Firstly, we define $f_{S}^\star(\hat{\bm{X}})$ as a specific detector of semantic information $S$ based on the reconstructed signal $\hat{\bm{X}}$, where $f_{S}^\star(\cdot):\prod_{i=1}^L\mathcal{X}_i\mapsto\mathcal{S}$ and it decides the reconstructed semantic information $\hat{S}=\ell$ when the sufficient statistic $\hat{X}_{s}\triangleq\frac{1}{L}\sum_{i=1}^L\hat{X}_i$ satisfies $\hat{X}_{s}\in \left[\ell-\frac{1}{2}, \ell+\frac{1}{2}\right]$. Obviously, the detector $f_S^\star(\hat{\bm{X}})$ is not the optimal detector for $S$, which will definitely have a larger error probability than $f_S(\bm{U})$ does. Therefore, we have 
		\begin{align}
			\mathbb{P}\left\{S\neq f_S(\bm{U})\right\}
			\leq&\mathbb{P}\left\{S\neq f_{S}^\star(\hat{\bm{X}})\right\}\notag\\
			=&\sum_{\ell=1}^M\omega_\ell\mathbb{P}\left\{f_{S}^\star(\hat{\bm{X}})\neq\ell|S=\ell\right\}\nonumber\\	
			\overset{(o)}{=}&\sum_{\ell=1}^M\omega_\ell\mathbb{P}\left\{\left.\left\vert\hat{X}_{s}-\ell\right\vert\geq\frac{1}{2}\right\vert S=\ell\right\}\label{div},
		\end{align}
		where step $(o)$ follows the detection rule of $f_S^\star(\hat{\bm{X}})$. However, since there is no evidence to show that $\hat{X}_S$ is an unbiased estimation of $X_S\triangleq\frac{1}{L}\sum_{i=1}^LX_i$ conditioned on specific $S=\ell$, hence the probability of $\hat{X}_S$ exceeding the decision region $\left[\ell-\frac{1}{2}, \ell+\frac{1}{2}\right]$ can not be directly upper bounded by Chebyshev inequality. To complete our proof, here we introduce an auxiliary variable $\alpha$, where $0<\alpha<\frac{1}{2}$
		is utilized to characterize the deviation 
		of $X_S$ to its mean when given $S=\ell$. Specifically, sketched in Fig. \ref{region}, if the density of $p(x|S=\ell)$\footnote{With a slight abuse of notation, $x$ and $\hat{x}$ denote the realization of sufficient statistics $X_S$ and $\hat{X}_S$ for simplicity.}is concentrated within the small region of $\left[\ell-\alpha, \ell+\alpha\right]$, then the probability of $\hat{X}_S$ exceeding the decision region $\left[\ell-\frac{1}{2}, \ell+\frac{1}{2}\right]$, namely the probability of $\hat{X}_S$ having a large deviation off the central of $p(\hat{x}|S=\ell)$, will be upper bounded or otherwise the MSE between $X_S$ and $\hat{X}_S$ will exceed the distortion constraint. Following this idea, we can split the event of $\hat{X}_S$ exceeding the decision region $\left[\ell-\frac{1}{2}, \ell+\frac{1}{2}\right]$ as
		\footnote{We remark that the decoder is not necessarily the same as $f_S^\star(\cdot)$, in which we fix the decision region and the form of sufficient statistic. Actually, we select a scalar estimator $\hat{X}_S$ and $X_S$ for easy analysis, while a high dimension decision region and estimator may tighten our result. }
		\begin{figure}[t]
			\centering
			\includegraphics[width=0.5\textwidth]{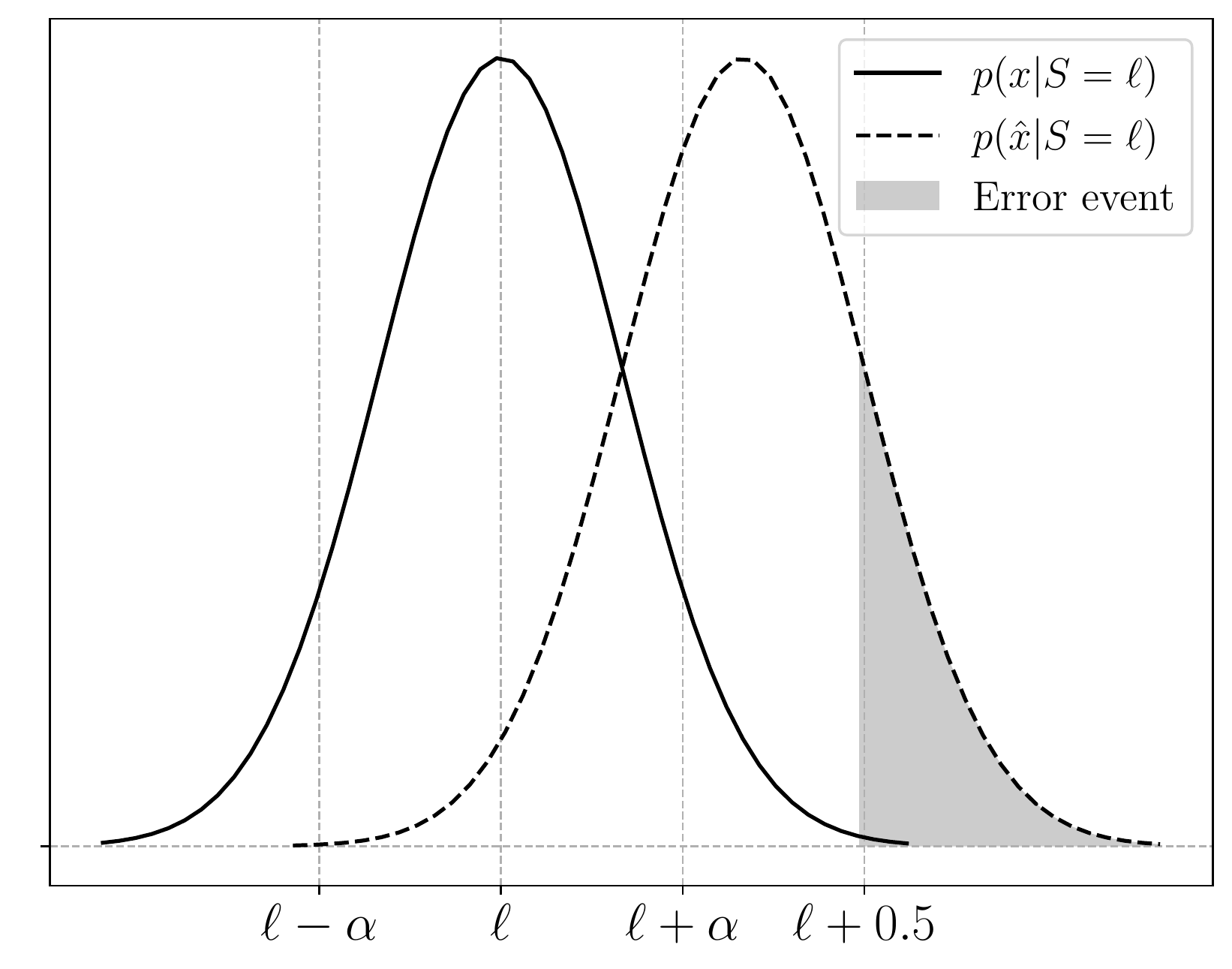}
			\caption{An illustration for the detection rule of specific detector $f_S^\star(\cdot)$}
			\label{region}
		\end{figure}
		\begin{align}
			&\mathbb{P}\left\{\left.\left\vert\hat{X}_S-\ell\right\vert\geq\frac{1}{2}\right\vert S=\ell\right\}\notag\\
			=&\mathbb{P}\left\{\left.\left\vert\hat{X}_S-\ell\right\vert\geq\frac{1}{2},\left\vert X_S-\ell\right\vert\geq\alpha\right\vert S=\ell\right\}+\mathbb{P}\left\{\left.\left\vert\hat{X}_S-\ell\right\vert\geq\frac{1}{2},\left\vert X_S-\ell\right\vert<\alpha\right\vert S=\ell\right\}\label{div2}.
		\end{align}
		Now for Eq. \eqref{div2}, we rewrite the first term as
		\begin{align}
			&\mathbb{P}\left\{\left.\left\vert\hat{X}_S-\ell\right\vert\geq\frac{1}{2},\left\vert X_S-\ell\right\vert\geq\alpha\right\vert S=\ell\right\}\notag\\
			=&\mathbb{P}\left\{\left.\left\vert\hat{X}_S-\ell\right\vert\geq\frac{1}{2}\right\vert \left\vert X_S-\ell\right\vert\geq\alpha,S=\ell\right\}\times\mathbb{P}\left\{\left\vert X_S-\ell\right\vert\geq\alpha|S=\ell\right\}\notag\\
			\overset{(p)}{\leq}&\mathbb{P}\left\{\left\vert X_S-\ell\right\vert\geq\alpha|S=\ell\right\}\\
			\overset{(q)}{=}&Q\left(\frac{L(\ell-\ell-\alpha)}{\sqrt{\mathrm{tr}\{\tb{K}_X\}}}\right)-Q\left(\frac{L(\ell-\ell+\alpha)}{\sqrt{\mathrm{tr}\{\tb{K}_X\}}}\right)\notag\\
			=&2Q\left(\frac{L\alpha}{\sqrt{\mathrm{tr}\{\tb{K}_X\}}}\right),\label{8}
		\end{align}
		where step $(p)$ follows the simple fact that 
		\begin{align}
			\mathbb{P}\left\{\left.\left\vert\hat{X}_S-\ell\right\vert\geq\frac{1}{2}\right\vert \left\vert X_S-\ell\right\vert\geq\alpha,S=\ell\right\}\leq1,
		\end{align}
		while step $(q)$ follows the definition of $Q$ function and the fact that $X_S\sim\mathcal{N}(\ell,\frac{\mathrm{tr}\{\tb{K}_X\}}{L^2})$ for given $S=\ell$. It should be noticed that introducing $\alpha$ can dynamically adjust the tightness of the bound. 
	
	For the second term, we consider the trace of covariance matrix of $\bm{X}$ given $\bm{U}$
	\begin{align}
		\frac{1}{L}\mathrm{tr}\left\{\tb{K}_{\bm{X}|\bm{U}}\right\}
		&=\frac{1}{L}\mathrm{tr}\left\{\mathbb{E}\left[\left(\bm{X}-\mathbb{E}[\bm{X}|\bm{U}]\right)\left(\bm{X}-\mathbb{E}[\bm{X}|\bm{U}]\right)^T\right]\right\}\notag\\
		&=\frac{1}{L}\sum_{i=1}^L\mathbb{E}\left[d_X(X_i,\hat{X}_i)\right]\notag\\
		&=\mathbb{E}\left[\frac{1}{L}\sum_{i=1}^Ld_X(X_i,\hat{X}_i)\right]\notag\\
		&\overset{(r)}{\geq}\mathbb{E}_{X_S,\hat{X}_S}\left[d_X(X_S,\hat{X}_S)\right]\notag\\
		&=\mathbb{E}_{X_S,\hat{X}_S,S}\left[d_X(X_S,\hat{X}_S)|S=\ell\right]\notag\\
		&=\mathbb{E}_{S}\left[\int ||x-\hat{x}||^2p(x,\hat{x}|S=\ell)dxd\hat{x}\right]\label{11},
	\end{align}
	where one can easily verify $(r)$ holds owing to the Jensen inequality and the concavity of MSE function. Moreover, we find the term within the expectation in Eq. \eqref{11} can be bounded as
	\begin{align}
		&\int ||x-\hat{x}||^2p(x,\hat{x}|S=\ell)dxd\hat{x}\notag\\
		=&\int_{x:|x-\ell|<\alpha}\int_{\hat{x}}||x-\hat{x}||^2p(x,\hat{x}|S=\ell)dxd\hat{x}+\int_{x:|x-\ell|\geq\alpha}\int_{\hat{x}}||x-\hat{x}||^2p(x,\hat{x}|S=\ell)dxd\hat{x}\notag\\
		\geq&\int_{x:|x-\ell|<\alpha}\int_{\hat{x}:|\hat{x}-\ell|\geq\frac{1}{2}}||x-\hat{x}||^2p(x,\hat{x}|S=\ell)dxd\hat{x}\notag\\
		&\hspace{6cm}+\int_{x:|x-\ell|<\alpha}\int_{\hat{x}:|\hat{x}-\ell|<\frac{1}{2}}||x-\hat{x}||^2p(x,\hat{x}|S=\ell)dxd\hat{x}\notag\\
		\geq&\int_{x,\hat{x}:|x-\ell|<\alpha,|\hat{x}-\ell|\geq\frac{1}{2}}||x-\hat{x}||^2p(x,\hat{x}|S=\ell)dxd\hat{x}\notag\\
		\overset{(s)}{\geq}&\left(\frac{1}{2}-\alpha\right)^2\int_{x,\hat{x}:|x-\ell|<\alpha,|\hat{x}-\ell|\geq\frac{1}{2}}p(x,\hat{x}|S=\ell)dxd\hat{x}\notag\\
		\overset{(t)}{=}&\left(\frac{1}{2}-\alpha\right)^2\mathbb{P}\left\{\left.\left\vert\hat{X}_S-\ell\right\vert\geq\frac{1}{2},\left\vert X_S-\ell\right\vert<\alpha\right\vert S=\ell\right\}\label{9},
	\end{align}
	where $(s)$ follows the fact that $||x-\hat{x}||^2\geq(\frac{1}{2}-\alpha)^2$ when $|x-\ell|<\alpha$ and $|\hat{x}-\ell|\geq\frac{1}{2}$. Moreover, $(t)$ holds due to the definition of the conditional probability. Now with the combination of  Eq. \eqref{11} and Eq. \eqref{9} we have
	\begin{align}
		&\mathbb{E}_S\left[\mathbb{P}\left\{\left.\left\vert\hat{X}_S-\ell\right\vert\geq\frac{1}{2},\left\vert X_S-\ell\right\vert<\alpha\right\vert S=\ell\right\}\right]\notag\\
		\leq&\left(\frac{1}{2}-\alpha\right)^{-2}\frac{1}{L}\mathrm{tr}\left\{\tb{K}_{\bm{X}|\bm{U}}\right\}\notag\\
		=&\frac{1}{L}\left(\frac{1}{2}-\alpha\right)^{-2}\mathrm{tr}\left\{\bm{\Gamma}\right\} \label{second},
	\end{align} 
	After that, with the substitution of Eq. \eqref{8} and Eq. \eqref{second} in Eq. \eqref{div}, where we pursue an optimal $\alpha$ minimizing the error probability, we get
	\begin{align}
		\mathbb{P}\left\{S\neq f_S(\bm{U})\right\}\leq\min\left\{1,2Q\left(\frac{L\alpha^*}{\sqrt{\mathrm{tr}\{\tb{K}_X\}}}\right)+\frac{1}{L}\left(\frac{1}{2}-\alpha^*\right)^{-2}\mathrm{tr}\left\{\bm{\Gamma}\right\}\right\}\label{12}\triangleq p_e.
	\end{align}
	where 
	\begin{align}
		\alpha^*\triangleq\mathrm{argmin}_{0<\alpha<\frac{1}{2}}\left\{2Q\left(\frac{L\alpha}{\sqrt{\mathrm{tr}\{\tb{K}_X\}}}\right)+\frac{1}{L}\left(\frac{1}{2}-\alpha\right)^{-2}\mathrm{tr}\left\{\bm{\Gamma}\right\}\right\}.
	\end{align}
	
	Note that $p_e$ can be obtained in such a non-trivial case with a high resolution regime, i.e. a small $\bm{\Gamma}$, which is illustrated in Fig. \ref{convexity}.
	\begin{figure}[t]
		\centering
		\includegraphics[width=0.55\textwidth]{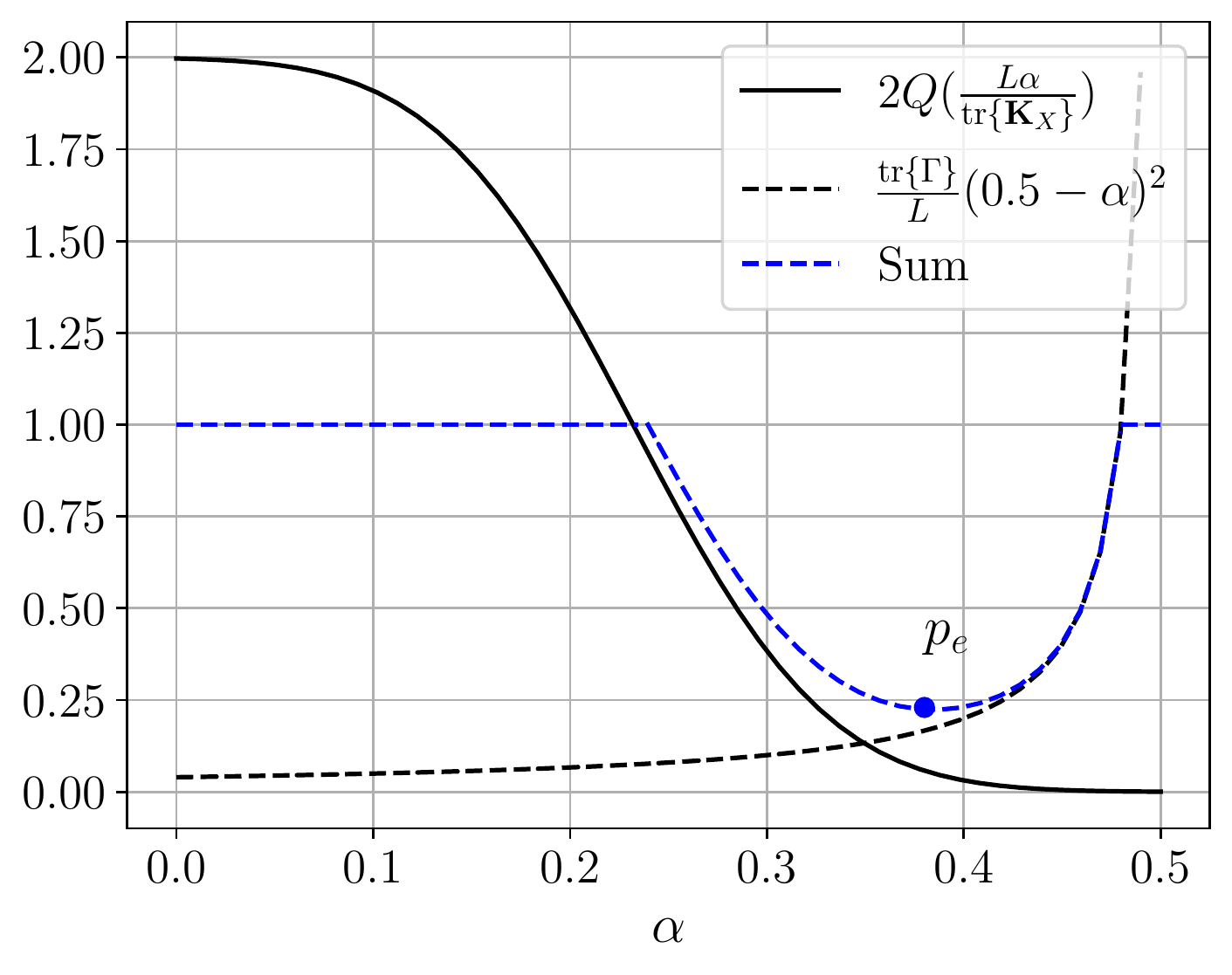}
		\caption{Convexity of $p_e$ over parameter $\alpha$}
		\label{convexity}
	\end{figure}
	Finally, we combine the above constraints in Eq. \eqref{beta_single user}, which yields
	\begin{align}
		R&\geq\min_{\beta,\bm{\Gamma}}H(\bm{\omega})-\beta+\frac{1}{2}\log_2\frac{\det(\tb{K}_X)}{\det(\bm{\Gamma})},\notag\\
		&\text{s.t.   }\beta\leq \min\left\{D_S,1+\log_2(M-1)p_e\right\},\\
		&\hspace{0.7cm}\bm{e}_i^T\bm{\Gamma}\bm{e}_i\leq D_{X_i}\quad\text{for}\quad i=1,2,\cdots,L.
	\end{align}
	where $p_e$ is formulated in Eq. \eqref{12}. Note that the upper bound of $\beta$ and error probability $p_e$ are both functions of $\bm{\Gamma}$, hence we reformulate the notations as $\beta(\bm{\Gamma})$ and $p_e(\bm{\Gamma})$, finally complete the proof of Theorem \ref{Theorem_label_Gaussian_Mixture}.
	
	\section{Proof of Corollary \ref{Contradiction}}\label{proof_contradiction}
	This corollary shows that the observed distortion constraint is always active for outer bound Eq. \eqref{Cout} when $0\leq D_{X_i}\leq \bm{e}_i^T\tb{K}_X\bm{e}_i$ for $i=1,2,\cdots,L$. This is an interesting results appearing with Gaussian mixture sources, while it is not the same case with joint Gaussian distributed sources \cite{Liu_Shao_Zhang_Poor_2022}. An intuitive interpretation of this phenomenon can be attributed to the difference between the finite and infinite alphabets in the Gaussian mixture source case. To prove this, we need to show the following two arguments:
	\begin{enumerate}
		\item $\bm{D}_X^\star$ exists,
		\item $D_S^\star$ does not exist,
	\end{enumerate}
	for outer bound Eq. \eqref{Cout} according to Def. \ref{dummy}.	To show 1), when fixing $D_S$, one can construct $\bm{D}_X^\star=(D^\star_{X_1},\cdots,D^\star_{X_L})$ satisfying
	\begin{align}
		1+\log(M-1)\left\{2Q\left(\frac{L\alpha}{\sqrt{\mathrm{tr}\{\tb{K}_X\}}}\right)+\frac{1}{L}\left(\frac{1}{2}-\alpha\right)^{-2}\sum_{i=1}^LD^\star_{X_i}\right\}= D_S,
	\end{align}
	such that for all $\bm{d}_X^\star=(d^\star_{X_1},\cdots,d^\star_{X_L})\prec\bm{D}_X^\star$,
	\begin{align}
		R^\mathrm{C}_\mathrm{out}(D_S,\bm{d}_X^\star)=&\min_{\bm{\Gamma}}H(\bm{\omega})-\left(1+\log_2(M-1)p_e(\bm{\Gamma})\right)+\frac{1}{2}\log_2\frac{\det(\tb{K}_X)}{\det(\bm{\Gamma})},\label{76}\\
		&\text{s.t. }\tb{O}\preceq\bm{\Gamma}\preceq\tb{K}_{X},\notag\\
		&\hspace{0.6cm} \bm{e}_i^T\bm{\Gamma}\bm{e}_i\leq d^\star_{X_i},\text{   for  } i=1,2,\cdots,L.\notag
	\end{align}
	where $p_e(\bm{\Gamma})$ is a function of $\mathrm{tr}\{\bm{\Gamma}\}$ defined in Eq. \eqref{28}. It should be noticed that Eq. \eqref{76} is only a function of $\bm{d}_X^\star$, namely the following equation is possible,
	\begin{align}
		R^\mathrm{C}_\mathrm{out}(D_S,\bm{d}_X^\star)=R^\mathrm{C}_\mathrm{out}(D_S+\Delta,\bm{d}_X^\star),
	\end{align}
	which validates the existence of $\bm{D}_X^\star$.
	
	To show 2), we can use reduction to absurdity, namely we need to show when fixing $D_{X_i}$, for arbitrary $D_S^\star$, there exists a unique semantic distortion constraint $d_S^\star<D_S^\star$ such that there is no positive vector $\bm{\Delta}$ satisfying 
	\begin{align}
		&R^\mathrm{C}_\mathrm{out}(d_S^\star,\bm{D}_X+ \bm{\Delta})=R^\mathrm{C}_\mathrm{out}(d_S^\star,\bm{D}_X).\label{64}
	\end{align}
	We start with the assumption that there exists a threshold $D_S^\star$ according Def. \ref{dummy}. Now recall the fact that 
	\begin{align}
		D_S\geq \mathbb{E}[d_S(S,\hat{S})|\bm{U}=\bm{u}]\geq H(S|\bm{U})\geq H(S|\bm{X}),
	\end{align}
	which means that a non-trivial semantic distortion $D_S$ should be lower bounded by $H(S|\bm{X})$ at the extreme case where $\bm{X}$ can be perfectly recovered by $\bm{U}$. Hence we let $d_S^\star=H(S|\bm{X})$, which satisfies  $d_S^\star< D_S^\star$ and $d_S^\star\leq H(S|\bm{U})$, in which our outer bound degenerates to
	\begin{align}
		R^\mathrm{C}_\mathrm{out}(d_S^\star,\bm{D}_X)=&\min_{\bm{\Gamma}}H(\bm{\omega})-H(S|\bm{X})+\frac{1}{2}\log_2\frac{\det\{\tb{K}_X\}}{\det\{\bm{\Gamma}\}},\\
		&\text{s.t. }\tb{O}\preceq\bm{\Gamma}\preceq\tb{K}_{X},\notag\\
		&\hspace{0.6cm} \bm{e}_i^T\bm{\Gamma}\bm{e}_i\leq D_{X_i},\text{   for  } i=1,2,\cdots,L.\notag
	\end{align}
	which can be rewritten as 
	\begin{align}
		R^\mathrm{C}_\mathrm{out}(d_S^\star,\bm{D}_X)=H(\bm{\omega})-H(S|\bm{X})+\frac{1}{2}\log_2\frac{\det\{\tb{K}_X\}}{\prod_{i=1}^LD_{X_i}}\label{65}.
	\end{align}
	since the Hardmard inequality $\det\{\bm{\Gamma}\}\leq\det\{\mathrm{diag}\{D_{X_1},\cdots,D_{X_L}\}\}$ \cite[Thm.~7.8.1]{Horn_Johnson_2012}. However, Eq. \eqref{65} is monotonically decreasing with $(D_{X_1},D_{X_2},\cdots,D_{X_L})$ over the intervals $0\leq D_{X_i}\leq \bm{e}_i^T\tb{K}_X\bm{e}_i$, i.e.,
	\begin{align}
		R^\mathrm{C}_\mathrm{out}(d_S^\star,\bm{D}_X)<R^\mathrm{C}_\mathrm{out}(d_S^\star,\bm{D}_X+\bm{\Delta})\label{69}.
	\end{align}
	Obviously Eq. \eqref{69} contradicts Eq. \eqref{64}, which means there exists no such positive vector $\bm{\Delta}$ in this counterexample, yielding the conclusion that we cannot find a threshold $D_S^\star$. 
	\section{Proof of Corollary \ref{Corollary_region}}\label{proof_corollary_classification}
	This corollary presents a detailed characterization for outer bound Eq. \eqref{Cout} based on the fact that the observed constraints are always active. The outline of this proof follows the categorical discussion where $D_S$ is active or not. To do this, we first argue the assumption of a diagonal $\bm{\Gamma}$ will not lose the generality, and assume an optimal $\alpha^\star$ by optimizing Eq. \eqref{28}, then formulate outer bound under different specific constraints. We first introduce the following lemma to argue that the optimal $\bm{\Gamma}$ in our outer bound is diagonal.
	\begin{lemma}\label{diagonal}
		For outer bound Eq. \eqref{Cout}, and non-diagonal $\bm{\Gamma}$ is sub-optimal, i.e.
		\begin{align}
				R^\mathrm{C}_{\mathrm{out}}(D_S,\bm{D}_X)=&\min_{\gamma_1,\cdots,\gamma_L\geq0}H(\bm{\omega})-\beta(\gamma_1,\cdots,\gamma_L)+\frac{1}{2}\log\frac{\det{\tb{K}_X}}{\prod_{i=1}^L\gamma_i}\\
				&\text{s.t.}\quad 0\leq\gamma_i\leq D_{X_i},\quad\text{for}\quad i=1,2,\cdots,L.
		\end{align}
	where $\beta(\gamma_1,\cdots,\gamma_L)$ is defined in Eq. \eqref{26b}, in which $\mathrm{tr}\{\bm{\Gamma}\}$ is substituted by $\sum_{i=1}^L\gamma_i$. 
	\end{lemma} 
	\begin{proof}
		Since the term $\beta(\bm{\Gamma})$ is either a constant or a function of $\mathrm{tr}\{\bm{\Gamma}\}$, the optimality of diagonal $\bm{\Gamma}$ can be obtained by the Hardmard inequality obviously, which tells us arbitrary positive definite $\Gamma\in\mathbb{R}^{L\times L}$ is inferior to its diagonal counterpart in terms of its determinant and trace, i.e.
		\begin{align}
			&\log\det\{\bm{\Gamma}\}\leq\log\det\{\bm{\Gamma}_D\}\\
			\text{and}\quad&\mathrm{tr}\{\bm{\Gamma}\}+\log\det\{\bm{\Gamma}\}\leq\mathrm{tr}\{\bm{\Gamma}_D\}+\log\det\{\bm{\Gamma}_D\}
		\end{align}
		where $\bm{\Gamma}_D\triangleq\mathrm{diag}\{\gamma_1,\cdots,\gamma_L\}$ and $\gamma_1,\cdots,\gamma_L$ is the main diagonal entries of $\bm{\Gamma}$.
	\end{proof}
	Moreover, we find that the optimal $\alpha$ is the function of $\gamma_1,\cdots,\gamma_L$. For further use, we claim there exists an optimal $\alpha^\star$ if it satisfies the following conditions:
	\begin{align}
		&\zeta\geq0,\notag\\
		&\zeta\left(\frac{1}{2}-\alpha\right)=0,\notag\\
		&\frac{L}{\sqrt{\mathrm{tr}\{\tb{K}_X\}}}\mathcal{N}\left(\frac{L\alpha}{\sqrt{\mathrm{tr}\{\tb{K}_X\}}};0,1\right)-\frac{\sum_{i=1}^L\gamma_i}{L}\left(\frac{1}{2}-\alpha\right)^{-3}-\zeta=0,\label{KKT}
	\end{align}
	where $-\mathcal{N}\left(x;0,1\right)$ is the derivative of Gaussian error function $Q(x)$. Now equipped with lemma \ref{diagonal}, given optimal $\alpha^*$ w.r.t arbitrary diagonal $\bm{\Gamma}=\mathrm{diag}\{\gamma_1,\cdots,\gamma_L\}$ and define $p_e(\sum_{i=1}^L\gamma_i)$ as Eq. \eqref{37}, the key to solve the optimization problem is to find the exact valve of  $\beta(\gamma_1,\cdots,\gamma_L)$, which  only takes values from the set 
	\begin{align}
		\left\{D_S,1+\log_2(M-1),1+\log_2(M-1)p_e\left(\sum_{i=1}^L\gamma_i\right)\right\}.\label{93}
	\end{align} 
	Consequently, we wish to find the minimum in Eq. \eqref{93}, which motivates us to divide the region $(D_S,\bm{D}_X)$ into:
	\begin{enumerate}[A)]
		\item subregion $\mathcal{R}_A$ with \textbf{active semantic distortion} and \textbf{low resolution observations}, namely satisfying
		\begin{align}
			H(S|\bm{X})\leq D_S&\leq 1+\log_2(M-1),\\
			H(S|\bm{X})\leq D_S&\leq 1+\log_2(M-1)p_e\left(\sum_{i=1}^LD_{X_i}\right).\label{95}
		\end{align}
		 Now if $D_S\leq 1+\log_2(M-1)p_e\left(\sum_{i=1}^L\gamma_i\right)$, then outer bound is reduced to 
		 \begin{subequations}
			\begin{align}
			R^\mathrm{C}_{\mathrm{out}}(D_S,\bm{D}_X)=&\min_{\gamma_1,\cdots,\gamma_L\geq0}H(\bm{\omega})-D_S+\frac{1}{2}\log\frac{\det\{\tb{K}_X\}}{\prod_{i=1}^L\gamma_i}\\
			&\text{s.t.}\quad 0\leq\gamma_i\leq D_{X_i},\quad\text{for}\quad i=1,2,\cdots,L,\\
			&\hspace{0.8cm}\sum_{i=1}^L\gamma_i\geq p_e^{-1}\left(\frac{D_S-1}{\log_2(M-1)}\right).
			\end{align}\label{96}
		\end{subequations}
		where $p_e^{-1}$ denotes the inverse function of $p_e(\cdot)$. Note that in this case, the non-empty feasible region for $\gamma_i$ is ensured by Eq. \eqref{95}. Therefore, the optimization problem is degenerated to the Gaussian rate-distortion function with semi-definite positive constraint only, namely
			\begin{align}
				R^\mathrm{C}_{\mathrm{out}}(D_S,\bm{D}_X)=H(\bm{\omega})-D_S+\frac{1}{2}\log\frac{\det\{\tb{K}_X\}}{\prod_{i=1}^LD_{X_i}},\label{97}
			\end{align}
			with $\gamma_i=D_{X_i}$ for all $i$. 
		\\if $D_S\geq 1+\log_2(M-1)p_e(\sum_{i=1}^L\gamma_i)$, then outer bound is reduced to 
\begin{subequations}
	\begin{align}
		R^\mathrm{C}_{\mathrm{out}}(D_S,\bm{D}_X)=&\min_{\gamma_1,\cdots,\gamma_L\geq0}H(\bm{\omega})-\left(1+\log_2(M-1)p_e\left(\sum_{i=1}^L\gamma_i\right)\right)+\frac{1}{2}\log\frac{\det\{\tb{K}_X\}}{\prod_{i=1}^L\gamma_i}\\
		&\text{s.t.}\quad 0\leq\gamma_i\leq D_{X_i},\quad\text{for}\quad i=1,2,\cdots,L,\\
		&\hspace{0.8cm}\sum_{i=1}^L\gamma_i\leq p_e^{-1}\left(\frac{D_S-1}{\log_2(M-1)}\right).
	\end{align}\label{98}
\end{subequations}
	It turns to an optimization problem with both semi-definite positive constraints and trace constraint and the optimal solution $\gamma_i=\frac{1}{L}p_e^{-1}\left(\frac{D_S-1}{\log_2(M-1)}\right)$ for all $i$, hence we conclude that in this case
	\begin{align}
		R^\mathrm{C}_{\mathrm{out}}(D_S,\bm{D}_X)= H(\bm{\omega})-D_S+\frac{1}{2}\log\det\{\tb{K}_X\}-\frac{L}{2}\log p_e^{-1}\left(\frac{D_S-1}{\log_2(M-1)}\right).\label{101}
	\end{align}
By comparing Eq. \eqref{97} and Eq. \eqref{101}, we find the former is always less than the later, hence we obtain the characterization $R^\mathrm{C}_{\mathrm{out}}(D_S,\bm{D}_X)$ as Eq. \eqref{97} with feasible region $\mathcal{R}_A$.

\item subregion $\mathcal{R}_B$ with \textbf{active semantic distortion} and \textbf{high resolution observations}, namely satisfying
\begin{align}
	H(S|\bm{X})\leq D_S&\leq 1+\log_2(M-1),\\
	D_S&\geq 1+\log_2(M-1)p_e\left(\sum_{i=1}^LD_{X_i}\right).\label{95}
\end{align}
we only have the possibility $D_S\geq 1+\log_2(M-1)p_e\left(\sum_{i=1}^L\gamma_i\right)$ due to $\gamma_i\leq D_{X_i}$, then outer bound takes the same form as optimization problem Eq. \eqref{98} with semi-definite positive constraint only, consequently we have
\begin{align}
	R^\mathrm{C}_{\mathrm{out}}(D_S,\bm{D}_X)=H(\bm{\omega})-\left(1+\log_2(M-1)p_e\left(\sum_{i=1}^LD_{X_i}\right)\right)+\frac{1}{2}\log\frac{\det\{\tb{K}_X\}}{\prod_{i=1}^LD_{X_i}},\label{100}
\end{align}
with $\gamma_i=D_{X_i}$ for all $i$. 
\item subregion $\mathcal{R}_C$ with \textbf{inactive semantic distortion} and \textbf{low resolution observations}, namely satisfying
\begin{align}
	D_S&\geq 1+\log_2(M-1),\\
	1&\leq p_e\left(\sum_{i=1}^LD_{X_i}\right) .\label{104}
\end{align}
Now if $p_e\left(\sum_{i=1}^L\gamma_i\right)\geq1$, then outer bound is reduced to 
\begin{subequations}
	\begin{align}
		R^\mathrm{C}_{\mathrm{out}}(D_S,\bm{D}_X)=&\min_{\gamma_1,\cdots,\gamma_L\geq0}H(\bm{\omega})-\left(1+\log_2(M-1)\right)+\frac{1}{2}\log\frac{\det\{\tb{K}_X\}}{\prod_{i=1}^L\gamma_i}\\
		&\text{s.t.}\quad 0\leq\gamma_i\leq D_{X_i},\quad\text{for}\quad i=1,2,\cdots,L,\\
		&\hspace{0.8cm}\sum_{i=1}^L\gamma_i\geq p_e^{-1}\left(1\right).
	\end{align}\label{105}
\end{subequations}
in which the non-empty feasible region of $\gamma_i$ is ensured by Eq. \eqref{104}. Thus
\begin{align}
	R^\mathrm{C}_{\mathrm{out}}(D_S,\bm{D}_X)=H(\bm{\omega})-(1+\log_2(M-1))+\frac{1}{2}\log\frac{\det\{\tb{K}_X\}}{\prod_{i=1}^LD_{X_i}}.\label{106}
\end{align}
Moreover, when $p_e\left(\sum_{i=1}^L\gamma_i\right)\leq1$, then 
\begin{subequations}
	\begin{align}
		R^\mathrm{C}_{\mathrm{out}}(D_S,\bm{D}_X)=&\min_{\gamma_1,\cdots,\gamma_L\geq0}H(\bm{\omega})-\left(1+\log_2(M-1)p_e\left(\sum_{i=1}^L\gamma_i\right)\right)+\frac{1}{2}\log\frac{\det\{\tb{K}_X\}}{\prod_{i=1}^L\gamma_i}\\
		&\text{s.t.}\quad 0\leq\gamma_i\leq D_{X_i},\quad\text{for}\quad i=1,2,\cdots,L,\\
		&\hspace{0.8cm}\sum_{i=1}^L\gamma_i\leq p_e^{-1}\left(1\right).
	\end{align}\label{107}
\end{subequations}
which is again an optimization problem with both semi-definite positive constraint and trace constraint. The optimal solution $\gamma_i=\frac{1}{L}p_e^{-1}\left(1\right)$ for all $i$, hence we conclude that in this case
\begin{align}
	R^\mathrm{C}_{\mathrm{out}}(D_S,\bm{D}_X)= H(\bm{\omega})-(1+\log_2(M-1))+\frac{1}{2}\log\det\{\tb{K}_X\}-\frac{L}{2}\log p_e^{-1}\left(1\right).\label{1088}
\end{align}
Obviously Eq. \eqref{106} is always less than Eq. \eqref{1088}, hence we obtain the characterization $R^\mathrm{C}_{\mathrm{out}}(D_S,\bm{D}_X)$ as Eq. \eqref{106} for feasible region $\mathcal{R}_C$.
\item the final subregion $\mathcal{R}_D$ with \textbf{inactive semantic distortion} and \textbf{high resolution observations}, namely satisfying
\begin{align}
	D_S&\geq 1+\log_2(M-1),\\
	1&\geq p_e\left(\sum_{i=1}^LD_{X_i}\right).\label{95}
\end{align}
Similarly, we only have the possibility $p_e\left(\sum_{i=1}^L\gamma_i\right)\leq1$. The outer bound takes the same form as Eq. \eqref{107} and we obtain
\begin{align}
	R^\mathrm{C}_{\mathrm{out}}(D_S,\bm{D}_X)=H(\bm{\omega})-\left(1+\log_2(M-1)p_e\left(\sum_{i=1}^LD_{X_i}\right)\right)+\frac{1}{2}\log\frac{\det\{\tb{K}_X\}}{\prod_{i=1}^LD_{X_i}},\label{}
\end{align}
with $\gamma_i=D_{X_i}$ for all $i$. Note that $R^\mathrm{C}_{\mathrm{out}}(D_S,\bm{D}_X)$ takes the same value on region $\mathcal{R}_B$ and $\mathcal{R}_D$, hence we combine these two regions as $\mathcal{R}_B^*=\mathcal{R}_B\bigcup\mathcal{R}_D$.
	\end{enumerate}
Finally, by incorporating the above feasible solutions and corresponding feasible regions, we reach the corollary \ref{Corollary_region}.
	\section{Proof of Theorem \ref{Theorem_label_coding_scheme}}\label{proof_coding_scheme}
	This appendix provides the proof of an inner bound for our semantic-aware MT problem with Gaussian mixture sources. The outline is the following:
	
	We devise a coding scheme wherein we initially identify the observations and estimate their semantic labels. These labels are then utilized in the quantization of Gaussian mixture observations. Furthermore, we employ entropy coding to compress both the labels and the quantized observations, reconstructing them at the decoder. The proof relies on the fact that our semantic-aware MT problem can be decomposed into a CEO problem concerning the semantic source and multiple compression problems for Gaussian sources with shared dependency, which will inevitably affect the performance. However, we will show the dependency can be eliminated by codebook construction theoretically.
	
	Recall the setups that $L=2$ and $M=2$, we start from the inner bound characterization in Sec. \ref{Sec3} by dividing the codebook into $C_i=[C_{i,1},C_{i,2}]$ for $i=1,2$, where $C_{i,1}$ denotes the quantized bits of semantic estimation and $C_{i,2}$ represents the quantized bits for Gaussian symbols. Thus the achievable rate $R_\mathrm{ach}\triangleq \sum_{i=1}^2\frac{1}{k}\log|C_i|$ can be written as
	\begin{align}
		R_\mathrm{ach}\geq&\sum_{i=1}^2H(C_i)\notag\\
		=&\sum_{i=1}^2 I(X_i;C_{i,1},C_{i,2})\notag\\
		\overset{(u)}{=}&\sum_{i=1}^2 I(S_i,X_i;C_{i,1},C_{i,2})\notag\\
		=&\sum_{i=1}^2 \left(I(S_i;C_{i,1},C_{i,2})+ I(X_i;C_{i,1},C_{i,2}|S_i)\right)\notag\\
		\overset{(v)}{\geq}&\sum_{i=1}^2\left(I(S_i;C_{i,1})+I(X_i;C_{i,2}|S_i)\right)\label{85}
	\end{align}
	where $(u)$ follows the fact that $S_i$ is the function of $X_i$, and $(v)$ follows the non-negativity of mutual information. Now the achievable rate is divided into rates for semantic estimation $I(S_i;C_{i,1})$ and rates for Gaussian quantization $I(X_i;C_{i,2}|S_i)$ for $i=1,2$, hence we aim to lower bound both terms in Eq. \eqref{85}, respectively. First, with $p_i=Q\left(1/\sigma_{X_i}\right)$, since $S$ is a symmetric Bernoulli source, $S_i\sim\mathrm{Ber}(\frac{1}{2}*p_i)$ is still a symmetric Bernoulli source. It should be noticed that the first term is reduced to the existing conclusion of Berger-Tung inner bound of binary sources and logarithmic loss in \cite[Eq.~(13)]{8425732}, i.e.
	\begin{align}
		\sum_{i=1}^2R_{i,1}\triangleq\sum_{i=1}^2 I(S_i;C_{i,1})+\epsilon_{i,1}
		=\min_{d_1,d_2} \left\{1+H_2(\mathscr{P}*\mathscr{D})-\sum_{i=1}^2H_2(d_i)+\epsilon_{i,1}\right\}
	\end{align}
	where $d_i\in(0,\frac{1}{2})$ denotes the flip probability of test channel $P_{C_i|S_i}$, and $\mathscr{P} = p_1*p_2$, $\mathscr{D} = d_1*d_2$. Since $D_S\geq H(S|C_1C_2)$, the optimized $d_i$ are subject to the logarithmic constraint as:
	\begin{align}
		H_2(p_1*d_1)+H_2(p_2*d_2)-H_2(\mathscr{P}*\mathscr{D})\leq D_S\label{87},
	\end{align}
	and $\epsilon_{i,1}$ denotes the error brought by coding scheme while it goes to $0$ with the increase of block length. 
	
	Then it is left to lower bound the second summation in Eq. \eqref{85}. It should be noticed that the conditional random variable $X_i|S_i=\ell$ for $i=1,2$ is not exactly independent Gaussian distributed, since $S_1$ and $S_2$ are both functions of $S$, thus it is necessary to investigate $I(X_i;C_{i,2}|S)$, i.e. the rate to quantize Gaussian signals with mandatory independent Gaussian codebook. Therefore, to measure the performance loss resulting from a mandatory independent Gaussian assumption, we compute the performance gap for $i=1,2$ as follows:
	\begin{align}
		\epsilon^\Delta_i &\triangleq I(X_i;C_{i,2}|S)-I(X_i;C_{i,2}|S_i)\notag\\
		&= I(SX_i;C_{i,2})-I(S;C_{i,2})-I(S_iX_i;C_{i,2})+I(S_i;C_{i,2})\notag\\
		&\overset{(x)}{=}I(S_i;C_{i,2})-I(S;C_{i,2})\notag\\
		&=H(C_{i,2}|S)-H(C_{i,2}|S_i)\notag\\
		&=H(C_{i,2}|S)-p_iH(C_{i,2}|S_i=S)-(1-p_i)H(C_{i,2}|S_i\neq S)\notag\\
		&=(1-p_i)\left[H(C_{i,2}|S_i=S)-H(C_{i,2}|S_i\neq S)\right]\label{108}
	\end{align}
	where $(x)$ follows that $I(SX_i;C_{i,2})=I(S_iX_i;C_{i,2})=I(X_i;C_{i,2})$. From Eq. \eqref{108}, the achievability can be intuitively obtain by constructing a codebook for Gaussian quantization and Slepian Wolf coding such that the codebook size remains invariant no matter how $S_i$ varies, namely $H(C_{i,2}|S_i=S)= H(C_{i,2}|S_i\neq S)$, then the gap between a mandatory Gaussian codebook and ideal codebook converges to 0 with the increasing block length. Thus we proceed to write
	\begin{align}
		R_{i,2}&=I(X_i;C_{i,2}|S)+\epsilon^Q_{i}+\epsilon^\Delta_{i}\geq \frac{1}{2}\log_2\frac{\sigma^2_{X_i}}{D_{X_i}}+\epsilon^Q_{i}+\epsilon^\Delta_{i}
	\end{align} 
	while the quantize error $\epsilon^Q_{i}$ is introduced by the quantization step. Finally, using the notation $\epsilon=\sum_{i=1}^2(\epsilon^Q_{i}+\epsilon_{i,1}+\epsilon_{i,2})$, we denote $R^\star(D_S,\bm{D}_X)$ the combination of constraint Eq. \eqref{87} with
	\begin{align}
		R_\mathrm{ach} =\sum_{i=1}^2(R_{i,1}+R_{i,2}),
	\end{align}
	which is shown in Theorem \ref{Theorem_label_coding_scheme}.
	
	%
	%
	\small
	\bibliographystyle{IEEEtran}
	\bibliography{RATE_DISTORTION_TIT}
	
	%

	
	

\end{document}